\pdfoutput=1

\documentclass{jgaa-art}
\usepackage[utf8]{inputenc}
\usepackage[T1]{fontenc}

\usepackage{amsfonts}
\usepackage{enumitem}
\usepackage{mathtools}
\usepackage{algorithm}
\usepackage[noend]{algpseudocode}
\usepackage{mod}
\usetikzlibrary{arrows}
\usepackage{booktabs}
\usepackage[normalem]{ulem}

\usepackage[binary-units=true]{siunitx}
\usepackage{subcaption}
\usepackage{pdflscape}
\usepackage{afterpage}
\usepackage{lineno}

\usepackage{pgfplots}
\definecolor{papergrey}{RGB}{51, 51, 51}
\definecolor{paperred}{RGB}{222, 6, 26}
\definecolor{paperblue}{RGB}{31, 119, 180}
\definecolor{papergreen}{RGB}{79, 180, 67}
\definecolor{paperyellow}{RGB}{8, 81, 156}
\definecolor{paperorange}{RGB}{245, 121, 0}
\definecolor{paperpurple}{RGB}{148, 0, 211}
\definecolor{paperlightpurple}{RGB}{178, 177, 204}
\definecolor{paperdarkpurple}{RGB}{67, 45, 133}

\newcommand{\plotfile}[1]{
	\pgfplotstableread{#1}{\table}
	\pgfplotstablegetcolsof{#1}
	\pgfmathtruncatemacro\numberofcols{\pgfplotsretval-1}
	\pgfplotsinvokeforeach{1,...,\numberofcols}{
		\pgfplotstablegetcolumnnamebyindex{##1}\of{\table}\to{\colname}
		\addplot table [y index=##1] {#1}; 
		\addlegendentryexpanded{\colname}
	}
}

\pgfplotscreateplotcyclelist{paper-cycle}{
	{papergrey},
	{paperred},
	{papergreen},
	{paperblue},
	{paperorange},
	{paperpurple},
	{paperlightpurple},
	{paperdarkpurple}
}

\pgfplotsset{compat=newest}

\pgfplotstableset{col sep=comma} 

\pgfplotsset{
	compat=newest,
	width=\textwidth, 
	height=10cm,
	cycle list name=paper-cycle,
	xtick pos=left,
	ytick pos=left,
	label style={font=\normalsize},
	legend style={font=\tiny},
	novo style/.style={
		thick,
		ymin=0,
		xmin=1,
		xtick={1,3,...,15},
	},
	novo bar style/.style={
		cycle list name=paper-cycle,
		height=\textwidth,
		ymin=0
	}
}

\newtheorem{thm}{Theorem}[section]
\newtheorem{lem}[thm]{Lemma}

\newtheorem{definition}[thm]{Definition}
\newtheorem{property}[thm]{Property}

\newtheorem{problem}{Problem}

\newcommand{\Der}[3]{#1\xRightarrow{#2}#3}

\usepackage{todonotes}

\newcommand{\Rule}[1]{(#1 \xleftarrow{l}{} K \xrightarrow{r}{} R)}
\newcommand{\pmap}{\rightharpoonup}
\newcommand{\tug}[1]{\widetilde{\mathbf{#1}}}
\newcommand{\ug}[1]{\mathbf{#1}}
\newcommand{\PA}{PA}
\newcommand{\IPA}{PA-I}
\newcommand{\EDE}{EDE}
\newcommand{\SEDE}{EDE-S}
\newcommand{\SSEDE}{EDE-SS}
\DeclareMathOperator{\dom}{dom}
\DeclareMathOperator{\codom}{codom}

\DeclareMathOperator{\id}{id}

\algnewcommand\algorithmicforeach{\textbf{for each}}
\algdef{S}[FOR]{ForEach}[1]{\algorithmicforeach\ #1\ \algorithmicdo}

\algdef{SE}[VARIABLES]{Variables}{EndVariables}
   {\algorithmicvariables}
   {\algorithmicend\ \algorithmicvariables}
\algnewcommand{\algorithmicvariables}{\textbf{global variables}}

\tikzset{vGraph/.style={node distance=0.5cm}}
\tikzset{eMorphism/.style={->, >=stealth}}
\tikzset{eIsomorphism/.style={double, double equal sign distance}}
\tikzset{vDot/.style={draw, fill, minimum size=3pt, inner sep=0, circle, outer sep=0}}
\tikzset{vDot2/.style={vDot, fill=none}}
\tikzset{eDot/.style={}}
\tikzset{eDotMorphism/.style={->, >=stealth, dashed, red}}
\tikzset{smaller labels/.style={every label/.append style={minimum size=0, inner sep=0, font=\scriptsize}}}
\tikzset{gDot/.style={node distance=2.5em and 2.5em, smaller labels}}
\tikzset{gCat/.style={node distance=2em and 2em}}
\tikzset{vCat/.style={draw, minimum size=2.5em}}

\newcommand{\summaryRuleSpan}[2][]{%
	\begin{tikzpicture}[
		node distance=20pt,
		vertex/.style={draw},
		normal/.style={eMorphism}
	]
	\node[vertex,label=above:$L$](L)						{\includegraphics[#1]{#2_L}};
	\node[vertex,label=above:$K$](K)[right=of L]	{\includegraphics[#1]{#2_K}};
	\node[vertex,label=above:$R$](R)[right=of K]	{\includegraphics[#1]{#2_R}};
	\draw[normal](K) to (L);
	\draw[normal](K) to (R);
	\end{tikzpicture}%
}

\begin{document}

\HeadingAuthor{Andersen et al.}
\HeadingTitle{Efficient Modular Graph Transformation Rule Application}

\title{Efficient Modular Graph Transformation Rule Application}
\Ack{This work is supported by the Novo Nordisk Foundation grant NNF19OC0057834 and by the Independent Research Fund Denmark, Natural Sciences, grants DFF-0135-00420B and DFF-7014-00041.}

\authorOrcid[first]{Jakob L. Andersen}{jlandersen@imada.sdu.dk}{0000-0002-4165-3732}
\authorOrcid[first]{Rolf Fagerberg}{rolf@imada.sdu.dk}{0000-0003-1004-3314}
\authorOrcid[first,second]{Juri Kolčák}{juri.kolcak@gmail.com}{0000-0002-9407-9682}
\authorOrcid[first]{Christophe V.F.P. Laurent}{christophe@imada.sdu.dk}{0000-0002-9112-6981}
\authorOrcid[first]{Daniel Merkle}{daniel@imada.sdu.dk}{0000-0001-7792-375X}
\authorOrcid[first]{Nikolai Nøjgaard}{nojgaard@imada.sdu.dk}{0000-0002-7053-4716}

\affiliation[first]{Department of Mathematics and Computer Science, University of Southern Denmark, Odense, Denmark}
\affiliation[second]{Department of Systems Biology, Harvard Medical School, Boston, MA, USA}

\maketitle

\begin{abstract}
Graph transformation formalisms have proven to be suitable tools for the modeling of chemical reactions. They are well established in theoretical studies \cite{g2006fundamentals} and increasingly also in practical applications in chemistry \cite{benko2003graph, yadav2004potential}. The latter is made feasible via the development of programming frameworks which make the formalisms executable~\cite{andersen2016software, Harris2016bionetgen, boutillier2018kappa}.

The application of such frameworks to large networks of chemical reactions poses unique computational challenges due to the nature of the involved graphs, since
these often consist of many individual connected components.
While the existing methods for implementing graph transformations can be applied to such graphs,
the combinatorics of constructing the graph matches involved in the application of graph transformation rules quickly becomes a computational bottleneck when the size of the chemical reaction network grows.

In this contribution, we develop a new method of enumerating graph matches during graph transformation rule application. The method is designed to improve performance in such scenarios and is based on constructing graph matches in an iterative, component-wise fashion which allows redundant applications to be detected early and pruned.
We further extend the algorithm with an efficient heuristic based on local symmetries of the graphs,
which allow us to detect and discard isomorphic applications early.
Finally, we conduct chemical network generation experiments on real-life as well as synthetic data and compare against the state-of-the-art algorithm in the field.
\end{abstract}

\section{Introduction}
A reaction network is a set of chemical reactions describing how a collection of molecules interact.
Reaction networks are used in many areas of research, for instance in biochemical pathway modeling~\cite{kanehisa2012kegg}, in drug design~\cite{grom2016modelling}, and in the study of planetary atmospheres~\cite{yung1998photochemistry}.
Known networks are commonly made available through manually curated databases, of which the KEGG \cite{kanehisa2012kegg} collection
of metabolic networks is one example.
Networks stored in such databases are inherently subject to sampling bias:
reactions and molecules are only recorded if the database curator deems them of interest in some known chemistry being examined \cite{benko2003graph}.
If one wishes to explore unknown chemistries, e.g., when designing new enzymatic mechanisms \cite{andersen2021graph},
one may instead turn to the generation of reaction networks using rule-based approaches.

Formal rule-based methods for generating reaction networks were already being investigated as early as the 90's \cite{fontana1990algorithmic}.
However, it is only recently that practical application has been made possible via the development of computational frameworks such as Kappa~\cite{boutillier2018kappa}, BioNetGen~\cite{Harris2016bionetgen}, and MØD~\cite{andersen2016software}.
Of these, Kappa and BioNetGen employ a unique labeling of connections, suited
for an abstract representation of binding sites and entire molecules, useful in modeling biological networks.
The unique labeling scheme is, however, ill suited for the modeling
of reaction networks, where the complete structure of each molecule at the level of individual bonds and atoms is of interest.

To the best of our knowledge, MØD~\cite{andersen2016software} is the only available computational framework with atom-level modeling of molecules. Indeed, its modeling is well aligned with the standard textbook description of molecules as chemical graphs, i.e., undirected graphs with vertices labeled by atom types and edges labeled by bond types. As argued in \cite{compSyst:17}, this strikes a good balance between chemical expressiveness and computational feasibility when studying reaction networks in a generative manner. We use the term chemical graph and the term molecule interchangeably from now on.

In MØD, the generation of reaction networks is achieved by means of rules for transforming chemical graphs.
A rule has a left side chemical graph pattern which specifies the minimum necessary molecular context for the transformation to take place. The rule can be applied to any molecule (or combination of molecules, for patterns with multiple parts) that contains the pattern. The right hand side of the rule then specifies how to transform the part of the molecule matching the pattern. Thus, applying a rule means matching its left side pattern to a molecule and then executing the transformation specified. In simple terms, a rule models a chemical reaction (or a set of chemical reactions with a common core, depending on the extent of the context specified) and applying the rule executes the reaction.
Constructing chemical networks is done by repeated application of the rules to any combination of the molecules generated, starting from some initial set of molecules.
The formalism underpinning MØD is the so-called double pushout (DPO) approach to algebraic graph transformation. Its use in MØD has been explored in multiple publications \cite{strat:14, andersen2017chemical, behr2019compositionality}.

A graph transformation rule application is a computationally intensive task as
it consists of monomorphism enumeration between the left side graph pattern of the rule and the host graph (i.e., the educt molecules of the reaction executed by the rule).
For larger reaction networks, the sheer number of rule and host graph
combinations makes the efficiency of the rule application of paramount
importance.

The nature of the challenge lies in identifying combinations of molecules which
constitute a host graph allowing a rule to be applied.
One way to do this is to compose the host graph and check the rule conditions at the same time.
This approach is employed in the state-of-the-art
method~\cite{andersen2016software}, which iterates over the connected components of the left side graph (i.e., the pattern) of the rule, while using partial rule application (PA) \cite{trace:14,rcMatch}.
The partial application of the rule is carried out by replacing a component of the rule left side graph by the host graph molecule they match.
The full rule application is then achieved by repeated partial application, each time generating a new version of the rule, until all left side components are replaced.

The number of rules generated by this partial application approach is in the worst case
exponential in the number of connected components of the left side graph of the rule.
Moreover, if the left side graph of the rule and the host graph are highly symmetric, an exhaustive
enumeration will invariably lead to redundant isomorphic applications.
Ideally, the symmetries of both the rule left side and host graph
should be taken into account when enumerating rule applications. However, 
since the rule is modified during the partial application, keeping track of such symmetries is difficult.
While the rules can be pruned by pairwise isomorphism checking, such an operation is computationally rather expensive.
This is further exacerbated when combined with the partial application approach,
where the generated rules grow in size as more molecules are matched onto the left side components.
For complex chemical networks with higher connected component counts and larger molecules,
e.g., enzyme mechanism networks~\cite{andersen2021graph}, the current algorithmic approach therefore struggles.

To address the combinatorial issues highlighted above, we in this paper develop a new algorithm for the
computation of all possible rule applications in a given set of molecules.
Similarly to the partial application approach, our algorithm relies on iterating over
the connected components of the rule left side graph.
However, instead of generating new versions of the rule, we explicitly construct the match (monomorphism)
of the rule into the host graph.
The match allows linking the components of the left side graph of the rule and the host molecules without
directly modifying them, thus avoiding the need to create copies of the rules.
Furthermore, we can also identify local symmetries of both the molecules and the rule and exploit them to prune partial matches known to lead to the same reactions.

We illustrate the speedup obtainable by this direct match construction versus the partial application approach in several experiments
including both synthetic and chemical data.
The results on synthetic data illustrate the potential for an exponential speedup with respect to the number of rule left side components.
A further improvement is then observed as a result of on-the-fly symmetry pruning.
The results on chemical data demonstrate that a speedup is achieved across wide variety of examples, even in cases with few rule left hand side components and close to no symmetries.

\section{Preliminaries}
In this contribution, we consider finite labeled simple graphs $G = (V, E, \lambda)$, with
vertex set $V = V(G)$, edge set $E = E(G)$ and labeling function
$\lambda\colon V \cup E \rightarrow \mathcal{L}$ over some set $\mathcal{L}$ equipped with an equality relation.

We assume the vertex set $V(G)$ is ordered by an arbitrary total order,
represented by an index set $I = \{1, \dots, |V(G)|\}$.  By abuse of notation,
we unify each vertex $v_i\in V(G)$ with its corresponding index $i$.  This is
in line with the practical representation of graphs where vertices explicitly
or implicitly have such an index.

Following the usual convention, given two (unlabeled) graphs $G_1 = (V_1,
E_1)$, $G_2 = (V_2, E_2)$,
a graph (homo)morphism is a function $\varphi\colon G_1 \rightarrow  G_2$, with domain and codomain
$\dom(\varphi) = V_1$ and $\codom(\varphi)\subseteq V_2$,
which preserves edges.
That is, for all $(u,v)\in E_1$ we have $(\varphi(u), \varphi(v))\in E_2$.
If the two graphs are labeled, i.e., $G_1 = (V_1, E_1, \lambda_1)$, $G_2 = (V_2, E_2, \lambda_2)$,
then $\varphi$ must further preserve labels, i.e., for any $v \in V_1$ we have
$\lambda_1(v) = \lambda_2(\varphi(v))$.

An injective graph morphism $\varphi$ is called a monomorphism,
and if one such exists we say that $G_1$ \emph{embeds} in $G_2$.
A monomorphism $\varphi$ that is also bijective is called an isomorphism,
and if one such exists we say that $G_1$ and $G_2$ are \emph{isomorphic},
written as $G_1 \simeq G_2$.
If an isomorphism is from a graph to itself, it is called an
\emph{automorphism}, and represents a symmetry of the graph.

For all of the morphism types we may deal with \emph{partial} versions of them. Notationally we write them with a harpoon, e.g.,
$\varphi\colon G_1 \pmap  G_2$.
For such morphisms we restrict the domain and codomain to only include
vertices of $G_1$ and $G_2$ that $\varphi$ is defined on.

In later sections, we build graphs by combining connected graphs. For this purpose, we utilize the following construction.

\begin{definition}[Union Graph]
    A \emph{union graph} $\mathbf{G}$ of dimension $n$ is a vector $(G_1, \dots, G_n)$ of connected graphs.

    For any graph $G$, let $\mathbf{G}(G) = \{i \in \{1, \dots, n\} \mid G_i \simeq G\}$
    be the set of indices of $\mathbf{G}$ which correspond to graphs isomorphic to $G$.

    The union graph $\mathbf{G}$ itself specifies a graph $\widetilde{\mathbf{G}} = \bigcup_{i = 1}^n G_i$ defined as the disjoint union of the \emph{component} graphs in $\mathbf{G}$.
    Additionally, each vertex $v \in V(\widetilde{\mathbf{G}})$ of the union graph can be uniquely mapped to a vertex of one of the component graphs of $\mathbf{G}$.
    We formalize the mapping as a bijection $\gamma\colon V(\widetilde{\mathbf{G}}) \rightarrow \bigcup_{i \in \{1, \dots, n\}} (\{i\}\times V(G_i))$.
    By abuse of notation, we write $v = (i, w)$ for any $v \in V(\widetilde{\mathbf{G}})$ such that $\gamma\colon v \mapsto (i, w)$.
    \label{def:union-graph}
\end{definition}

\begin{definition}[Union Graph Extension]
    Given a union graph $\mathbf{G} = (G_1, \dots, G_n)$ and a graph $G$,
	the \emph{extension} of $\mathbf{G}$ with $G$ is the union graph:
    \[
		\mathbf{G} \cdot G = (G_1, \dots, G_n, G)
	\]

    \label{def:graph_extension}
\end{definition}

For the transformation of a graph into another, we use the Double Pushout (DPO) approach \cite{corradini1997algebraic}. For a detailed overview of the different variations of
DPO we refer to \cite{habel2001double}. For this paper, we use DPO specifically as defined in~\cite{andersen2016software}
which is aimed at using DPO to model chemistry.
In this framework, a rule $p = \Rule{L}$
consists of a left graph $L$, a context graph $K$, and a right graph $R$, as well as monomorphisms $l$ and $r$ describing how $K$ is embedded in $L$ and
$R$. The transformation captured by a rule consists of replacing the graph $L$ with the graph $R$, while preserving the common parts, identified by $K,l$ and $r$. Note that $L$, $K$, and $R$ are not necessarily connected graphs. 

The application of a rule $\Rule{L}$ to a graph $G$ requires the existence of a morphism $m\colon L \rightarrow G$,
which we, as per~\cite{andersen2016software}, also require to be a monomorphism.
We refer to such a morphism as a \emph{match} of rule $p$ into $G$.
However, a match $m$ alone is not enough to guarantee that the rule can be applied,
as the resulting structure has to be a valid graph.
We therefore require the usual gluing conditions of DPO transformation
\cite[Def.\ 3.9]{g2006fundamentals}, namely
the \emph{identification condition} and the \emph{dangling condition}.
However, as $m$ is injective the \emph{identification condition} is trivially fulfilled.
In addition we add a \emph{parallel edge condition}, as we require the graphs to be simple.
We thus arrive at:
\begin{definition}[Valid Match]
    Let $p = \Rule{L}$ be a rule, let $G$ be a graph, and let $m\colon L \rightarrow G$ be a match of $p$ into $G$.

    Then the match $m$ is valid if the following conditions are satisfied:
    \begin{enumerate}[label=(\arabic*)]
	\item \textbf{Dangling Condition}:
		If a vertex $u \in V(L)$ is removed, $u \notin \codom(l)$, all the incident edges must also be removed:
		For each edge $(m(u), v') \in E(G)$ there must be an edge $(u, v) \in E(L)$ with $m(v) = v'$.%
		\label{itm:no_dangling_edge}
	\item \textbf{Parallel Edge Condition}: For any pair of vertices $v, w \in V(K)$ that the rule creates an edge between,
		$(l(v), l(w)) \notin E(L)$ and $(r(v),r(w)) \in E(R)$,
		the graph $G$ may not have an edge between them, $(m(l(v)), m(l(w))) \notin E(G)$.%
            \label{itm:no_parallel_edge}
    \end{enumerate}

    \label{def:valid_match}
\end{definition}

We refer to the application of a rule $p$ on a graph $G$ with a valid match $m\colon L \rightarrow G$ as a \emph{direct derivation} and denote it by $\Der{G}{p,m}{H}$, where $H$ is again a labeled simple graph~\cite{andersen2016software}.

Given a union graph $\mathbf{G}$ such that $\widetilde{\mathbf{G}} = G$, and a derivation $\Der{G}{p,m}{H}$,the match $m$ is also valid for an arbitrary extension $\mathbf{G} \cdot G'$, giving us the derivation $\Der{\widetilde{\mathbf{G} \cdot G'}}{p,m}{\widetilde{(H)\cdot G'}}$.
However, the extension graph $G'$ is irrelevant for the actual graph transformation.
This leads us to define a notion of a minimal graph into which a match can be embedded.
This notion is captured by what is known as a \emph{proper derivation}
\cite{strat:14} and the corresponding \emph{proper match}.
A match is proper if every connected component of the host graph $G$ intersects
the codomain of the match $m$.
A derivation is then proper if its match is proper.
We can express this notion very naturally using union graphs.

\begin{definition}[Proper Match]
    Let $\mathbf{G} = (G_1, \dots, G_n)$ be a union graph,
    let $\Rule{L}$ be a rule, and let $m\colon L \rightarrow \widetilde{\mathbf{G}}$ be a match.

    Then $m$ is proper if for each $i \in \{1, \dots, n\}$, there exist $v \in V(L)$ and $w \in V(G_i)$ such that $m(v) = (i, w)$.

    Given a valid proper match $m\colon L \rightarrow G$, the derivation $\Der{G}{p,m}{H}$ is called proper.
    \label{def:proper_match}
\end{definition}

\section{Enumerating Proper Derivations}
\label{sec:match_enumeration}

As stated in the introduction, the goal of this paper is to present an
algorithm to enumerate the possible direct derivations of a rule one can create by combining a set of given graphs. We can formalize this problem as follows.
\begin{problem}
	Let $\mathcal{G}$ be a finite set of pairwise non-isomorphic connected graphs and $p = \Rule{L}$ a rule.

	Enumerate all proper derivations
	\begin{align*}
	\Der{\widetilde{\mathbf{G}}}{p,m}{H}
	\end{align*}
	for all selections of cardinality $n\in \mathbb{N}$ and union graphs $\mathbf{G} \in \mathcal{G}^n$,
	but only up to reordering of the constituent graphs of each $\mathbf{G}$.
	\label{prob:derivations_for_graph_set}
\end{problem}
That is, we want to enumerate all proper derivations that can be obtained
by constructing the host graph as combinations of graphs in $\mathcal{G}$.
Clearly, the host graph obtained from a simple reordering of elements in a
combination of graphs in $\mathcal{G}$ is uninteresting, which is why we ignore such
reordering of the constituent graphs of each $\mathbf{G}$.
As we are only interested in proper derivations the cardinality $n$ is bounded by the number of connected components of $L$.

In the following we assume the left side graph $L$ of the rule $p$ corresponds to a union graph $L = \widetilde{\mathbf{L}}$.
This allows us to approach the problem by iterating over the component graphs $(L_1, \dots, L_k) = \mathbf{L}$
and accordingly extending the host graph $\mathbf{G}$.

More precisely, we utilize a notion of a \emph{partial match} which is simply a partial monomorphism of the rule left side into the host graph.
We iteratively extend a partial match with a monomorphism from a connected component of the left graph of a rule,
eventually obtaining a total match.
We formalize this in the notion of a \emph{partial match extension}, akin to the
union graph extension (Definition~\ref{def:graph_extension}).

\begin{definition}[Partial Match Extension]
	Let $\Rule{\widetilde{\mathbf{L}}}$ be a rule with $\mathbf{L} = (L_1, \dots, L_k)$, $\mathbf{G} = (G_1, \dots, G_n)$ a union graph,
	$i\in \{1,\dots, k\}$ and $x \in \{1, \dots, n\}$ two indices,
	and let $m\colon \widetilde{\mathbf{L}} \pmap \widetilde{\mathbf{G}}$ be a partial match.
	Finally, let $\varphi\colon L_i \rightarrow G_x$ be a monomorphism such
	that $L_i$ is undefined in $m$ and for all $v \in \codom(\varphi)$, $(x, v) \notin \codom(m)$.

Then the partial match extension is the partial match $m \mathbin{{\cup}_x^i} \varphi\colon \widetilde{\mathbf{L}} \pmap \widetilde{\mathbf{G}}$
of the rule $p$ into $\widetilde{\mathbf{G}}$ defined as follows:
\begin{align*}
	m \mathbin{{\cup}_x^i} \varphi = \left\{\begin{aligned}
		(i, v) &\mapsto (x, \varphi(v)) 	&\text{for all $v\in V(L_i)$}\\
		w &\mapsto m(w)					&\text{for all $w\in \dom(m)$}
	\end{aligned}\right.
\end{align*}

	To ease notation, we write simply $m \cup \varphi$ where the indices $i$ and $x$ are obvious from the context.

	\label{def:match_extension}
\end{definition}

The partial match extension thus lifts the monomorphism $\varphi$ from one of the component graphs into the union graph $\mathbf{G}$ itself.
The resulting partial morphism is always a partial monomorphism since $m$ is restricted to components other than $L_i$ and the codomain of $\varphi$ lifted to the union graph is disjoint with the codomain of $m$.
The match extension does not however guarantee that the resulting match
obtained from repeated extensions of a partial match is valid.

To ensure that any extension might potentially lead to a valid match we first
extend Definition~\ref{def:valid_match} of a valid match to partial matches.
This is done by limiting the scope of the definition to the domain of the match itself, indicated by the underlined sections.
Condition~\ref{itm:partial_no_dangling_edge} again corresponds to the dangling edge condition
while Condition~\ref{itm:partial_no_parallel_edge} corresponds to the parallel edge condition.

\begin{definition}[Valid Partial Match]
    Given a partial match $m\colon \widetilde{\mathbf{L}} \pmap G$ for a rule $p = \Rule{\widetilde{\mathbf{L}}}$ with $\mathbf{L} = (L_1, \dots, L_k)$
    on a graph $G$, we say that $m$ is valid if it satisfies the following
    criteria:
    \begin{enumerate}[label=(\arabic*)]
	\item \textbf{Dangling Condition}: For a vertex $u \in V(L)$
		\underline{defined in $m$} that is removed, $u \notin \codom(l)$,
		all the incident edges must also be removed:
		for each edge $(m(u), v') \in E(G)$ there must be an edge $(u, v) \in E(L)$ with $m(v) = v'$.%
		\label{itm:partial_no_dangling_edge}
	\item \textbf{Parallel Edge Condition}: For any pair of vertices $v, w \in
		V(K)$ that the rule creates an edge between,
		$(l(v), l(w)) \notin E(L)$ and $(r(v),r(w)) \in E(R)$,
		\underline{and are in the match, $l(v)$,} \underline{$l(w) \in \dom(m)$},
		the graph $G$ may not have an edge between them, $(m(l(v)), m(l(w))) \notin E(G)$.%
		\label{itm:partial_no_parallel_edge}
    \end{enumerate}

    \label{def:valid_partial_match}
\end{definition}

Clearly, any partial match obtained from a valid match is again valid.
As a result, we are interested in extensions that lead to valid partial matches.
\begin{definition}[Valid Partial Match Extension]
	Let $\Rule{\widetilde{\mathbf{L}}}$ be a rule with $\mathbf{L} = (L_1, \dots, L_k)$, $\mathbf{G} = (G_1, \dots, G_n)$ a union graph,
	$i\in \{1,\dots, k\}$ and $x \in \{1, \dots, n\}$ two indices,
	and let $m\colon \widetilde{\mathbf{L}} \hookrightarrow \widetilde{\mathbf{G}}$ be a valid partial match.
	Finally, let $\varphi\colon L_i \rightarrow G_x$ be a monomorphism such that $L_i$ is undefined in $m$.

	Then the extension $m' = m \cup_x^i \varphi$ is valid if it is defined and for any
	 vertex $v \in V(L_i)$ we have that $(i, v)\in V(\tug{L})$ satisfy the Dangling Condition of Definition \ref{def:valid_partial_match} in $m'$ and any pair $v,w$, where $w \in V(\tug{L})$, satisfy the Parallel Edge Condition of \ref{def:valid_partial_match} in $m'$.

	\label{def:valid_match_extension}
\end{definition}

Definition \ref{def:valid_match_extension} ensures that any valid partial match extension
is defined and produces another valid partial match by iterating over the vertices in $L_i$. The fact that $m'$ is valid results from the fact that $m$ is valid, and hence any
vertex not satisfying Definition \ref{def:valid_partial_match} must invariably involve
vertices of $L_i$. We note that for any extension to be valid, the corresponding monomorphism must itself be a valid partial match, and it thus suffices to only extend
these monomorphisms. However, even if both $m$ and $\varphi$ are valid partial matches,
their extension might not produce a valid partial match, as is depicted in Figure \ref{fig:invalid_partial_match}.

\begin{figure}
    \centering
    \begin{tikzpicture}[gCat, remember picture]
        \node[vCat, label=above:$\tug{L}$] (L) {
            \begin{tikzpicture}[smaller labels, remember picture]
              \node[modStyleGraphVertex, label={left:$(1, 1)$}] (l1) {C};
              \node[modStyleGraphVertex, label={right:$(2, 1)$}, right=of l1] (l2) {C};
            \end{tikzpicture}
        };
        \node[vCat, label=above:$K$, right=of L] (K) {
            \begin{tikzpicture}[remember picture]
              \node[modStyleGraphVertex] (k1) {C};
              \node[modStyleGraphVertex, right=of k1] (k2) {C};
            \end{tikzpicture}};
        \node[vCat, label=above:$R$, right=of K] (R) {
            \begin{tikzpicture}[remember picture]
              \node[modStyleGraphVertex] (r1) {C};
              \node[modStyleGraphVertex, right=of r1] (r2) {C};
              \draw[modStyleGraphEdge] (r1) to (r2);
            \end{tikzpicture}
        };
        \node[vCat, label=below:$\tug{G}$, below=of L] (G) {
            \begin{tikzpicture}[smaller labels, remember picture]
              \node[modStyleGraphVertex, label={left:$(1, 1)$}] (g1) {C};
              \node[modStyleGraphVertex, label={right:$(1, 2)$}, right=of g1] (g2) {C};
              \draw[modStyleGraphEdge] (g1) to (g2);
            \end{tikzpicture}
        };
        \node[vCat, label=below:$D\vphantom{\tug{G}}$, right=of G] (D) {
            \begin{tikzpicture}[remember picture]
              \node[modStyleGraphVertex] (d1) {C};
              \node[modStyleGraphVertex, right=of d1] (d2) {C};
              \draw[modStyleGraphEdge] (d1) to (d2);
            \end{tikzpicture}
        };
        \node[vCat, label=below:$H\vphantom{\tug{G}}$, right=of D] (H) {
            \begin{tikzpicture}[remember picture]
              \node[modStyleGraphVertex] (h1) {C};
              \node[modStyleGraphVertex, right=of h1] (h2) {C};
              \draw[modStyleGraphEdge, bend left=20] (h1) to (h2);
              \draw[modStyleGraphEdge, bend right=20] (h1) to (h2);
            \end{tikzpicture}
        };
        \path[eMorphism] (K) edge node [above] {$l$} (L);
        \draw[eMorphism] (K) edge node [above] {$r$} (R);
        \draw[eMorphism] (L) edge node [left] {$m$}  (G);
        \draw[eMorphism] (K) to (D);
        \draw[eMorphism] (R) to (H);
        \draw[eMorphism] (D) to (G);
        \draw[eMorphism] (D) to (H);

        \draw[eDotMorphism, bend right] (l1) to (g1);
        \draw[eDotMorphism, bend left] (l2) to (g2);
    \end{tikzpicture}%
    \caption{A rule application with the union graphs $\ug{L} = (L_1, L_2)$ and $\ug{G} = (G_1)$
    	corresponding to the connected components in $\tug{L}$ and $\tug{G}$ respectively.
    	The red arrows represent the two monomorphisms $\varphi_1\colon L_1 \rightarrow G_1$ and $\varphi_2\colon L_2 \rightarrow G_1$,
    	each individually a valid partial match.
    	The match $m = \varphi_1\cup_1\varphi_2$ and the corresponding final extension,
    	is however not valid because a parallel edge is created in $H$.
    	}
    \label{fig:invalid_partial_match}
\end{figure}
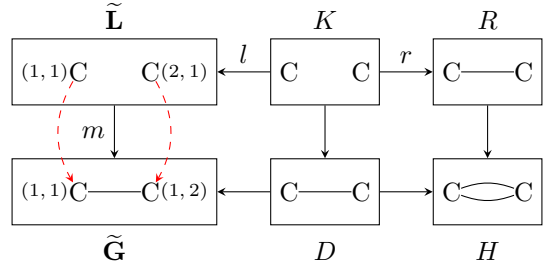

The partial match extension allows us to construct proper matches of a rule $p
= \Rule{\widetilde{\mathbf{L}}}$ into a graph $\widetilde{\mathbf{G}}$
iteratively over the component graphs of the left side union graph $\mathbf{L}
= (L_1, \dots, L_k)$.
Given a match $m\colon \tug{L} \rightarrow \tug{G}$, we let $m_i$
be the monomorphism mapping $L_i$ into some component graph of $\mathbf{G}$.

\begin{property}[Match Decomposition]
    Let $\mathcal{G}$ be a set of graphs and let $\Rule{\widetilde{\mathbf{L}}}$ be a rule with $\mathbf{L} = (L_1, \dots, L_k)$.
    Then, any proper match $m\colon \widetilde{\mathbf{L}} \rightarrow \widetilde{\mathbf{G}}$,
	where $\mathbf{G} \in {\mathcal{G}}^n$ can be rewritten as a
	sequence of valid partial match extensions $m = m_0 \cup m_1 \cup \dots \cup m_k$
	where $m_0\colon \widetilde{\mathbf{L}} \rightarrow \widetilde{\mathbf{G}}$ is the
	empty match of $\widetilde{\mathbf{L}}$ into $\widetilde{\mathbf{G}}$, $\dom(m_0) = \codom(m_0) = \emptyset$.

    \label{prop:match_decomposition}
\end{property}

We can always construct an empty match between any left side graph and any host graph. In what follows we thus assume the existence of the empty match implicitly and do not write it in the decomposition.

Following Property~\ref{prop:match_decomposition}, we know we can construct any proper match of a rule $\Rule{L}$ into a collection of graphs from a given set $\mathcal{G}$ by iteratively extending a partial match over the connected components of the left side graph $L$.
The iterative approach is useful in the context of Problem~\ref{prob:derivations_for_graph_set} as it allows us to iteratively construct a host union graph $\mathbf{G}$ alongside the match $m$.
Such iterative extension of the host graph restricts the possible order of the components of $\mathbf{G}$ by the order in which the connected components of the left side graph $L$ are matched.
However, the graph $\widetilde{\mathbf{G}}$ being a disjoint union of the component graphs,
any reordering of the components of $\mathbf{G}$ amounts only to a particular type of automorphism of $\widetilde{\mathbf{G}}$.

\begin{property}[Component Graph Order Independence]
	Let $\mathbf{G} = (G_1, \dots, G_n)$ and $\mathbf{H} = (H_1, \dots, H_n)$ be two union graphs
	such that there exists a bijection $\iota\colon \{1, \dots, n\} \rightarrow \{1, \dots, n\}$
	with $G_i = H_{\iota(i)}$ for all $i\in \{1, \dots, n\}$.
	Then, the two graphs $\widetilde{\mathbf{G}} \simeq \widetilde{\mathbf{H}}$ are isomorphic.

    \label{prop:order_independence}
\end{property}

Property~\ref{prop:order_independence} formalizes our earlier claim that we are only interested in one of the possible orderings of the components graphs in a union graph.
We can then solve Problem~\ref{prob:derivations_for_graph_set} by Algorithm~\ref{alg:enumerate_derivations}.
The algorithm enumerates valid matches recursively, interleaving match extensions with the host union graph $\mathbf{G}$ extensions.
Since each host graph extension is followed by a match extension into the new component, the constructed matches are guaranteed to be proper as long as each graph $G \in \mathcal{G}$ is connected.

The match extensions are always conducted in the order given by the component graphs $L_1, \dots, L_k$ of the left side graph of the rule.
The partial matches can then be arranged in a tree structure with the empty partial match $m_0\colon \emptyset \rightarrow \widetilde{()}$ in the root.
Then, any node on the $i$-$th$ level of the tree represents a partial match $m\colon \widetilde{\mathbf{L}} \rightarrow \widetilde{\mathbf{G}}$
for some $\mathbf{G} \in {\mathcal{G}}^n$ defined on $L_1, \dots, L_i$,
and any child of the node represents a possible extension $m \cup \varphi$, where $\varphi \colon  L_{i+1} \rightarrow G$ for some $G\in \mathcal{G}$.
The algorithm recursively navigates the above tree structure in a depth-first approach.
This is done by only keeping a single partial match and host union graph in memory, which we extend and shorten as necessary.

More precisely, the algorithm first computes all monomorphisms, that are also valid partial matches,
for each connected component $L_i\in \mathbf{L}$ (Line \ref{line:store_mono}).
The procedure $\textsc{EnumeratePartial}$ is then called recursively with the parameter $i$.
At each call, we check if $m$ is total, and hence a match, in which case we yield the corresponding direct derivation (Line \ref{line:yield_der}).
Otherwise, we enumerate all computed monomorphisms for $L_i$,
and for each $\varphi\colon L_i \rightarrow G$ we enumerate all possible positions of $G$ within $\mathbf{G}$ in which we try and extend $m$ (Line \ref{line:try_extend_1}).
In addition, we also enumerate the case where $\varphi$ is extended into a graph not yet in $\mathbf{G}$ (Line \ref{line:try_extend_2}).
Finally, in all cases, we remove $\varphi$ from $m$ when returning from the recursively called function, as well as removing $G$ from the end of $\mathbf{G}$ if it was extended.

\begin{thm}
Let $\mathcal{G}$ be a finite set of pairwise non-isomorphic graphs and $p = \Rule{L}$ a rule.
Then Algorithm~\ref{alg:enumerate_derivations} enumerates all proper derivations
as defined in Problem~\ref{prob:derivations_for_graph_set}.
\end{thm}

\begin{proof}
The algorithm constructs all possible combinations of the monomorphisms
$\varphi\colon L_i \rightarrow G$ for some $i \in \{1, \dots, k\}$ and $G \in
\mathcal{G}$, filtering them for validity.
As we are only interested in proper derivations, there is no need to extend the host graphs beyond the component graphs required by the monomorphisms themselves.
It is therefore easy to see that all proper matches of interest and the resulting derivations are enumerated at least once.

What remains to be shown is that each such match, respectively derivation, is
only enumerated for a single host graph.
The order of the component graphs of the host graphs constructed by Algorithm~\ref{alg:enumerate_derivations} is uniquely determined by the rule left side graph components and the match itself.
The order of the rule left side components being fixed, change in the order of the host graph components necessitates a change of the match itself.

The set $\mathcal{G}$ containing pairwise non-isomorphic graphs, two host graphs $\widetilde{\mathbf{G}} \simeq \widetilde{\mathbf{G'}} \in {\mathcal{G}}^n$ being isomorphic necessitates they only differ on the order of their component graphs, as the component graphs themselves have to be identical.
\end{proof}

\begin{algorithm}
\caption{Enumerate($\mathcal{G}, p=\Rule{\widetilde{\mathbf{L}} = (L_1, \dots, L_k)})$)}
 \label{alg:enumerate_derivations}
\begin{algorithmic}[1]
 \State $\mathbf{G} \gets$ the empty union graph
 \State $m \gets$ the empty partial match $\widetilde{\mathbf{L}} \pmap \widetilde{\mathbf{G}}$
    \ForEach{$i \in \{1, \dots, k\}$}
        \State $\mathcal{L}_i \gets \{\varphi\colon L_i \rightarrow G \mid G \in \mathcal{G} \wedge \varphi \text{ is a valid partial match}\}$ \label{line:store_mono}
        \If{$|\mathcal{L}_i|$ = 0}
            \State \Return
        \EndIf
    \EndFor
    \State $\textsc{EnumeratePartial}(1)$
 \State
 \Procedure{EnumeratePartial}{$i$}
    \If{$i = k+1$}\Comment{$m$ is total}
        \State \textbf{yield} $\Der{\widetilde{\mathbf{G}}}{p, m}{H}$ \label{line:yield_der}
        \State \Return
    \EndIf
    \ForEach{$\varphi\colon L_i \rightarrow G \in \mathcal{L}_i$}
        \ForEach{$x \in \{y\in\{1, \dots, n\}\ |\ G_y = G\}$}
            \If{$m \cup_x \varphi$ is valid}\label{line:try_extend_1}
            \Comment{try $\varphi$ on existing graphs}
            \State $m \gets m \cup_x \varphi$
                \State $\textsc{EnumeratePartial}(i+1)$
            \State remove $\varphi$ from $m$.
            \EndIf
        \EndFor
        \State $\mathbf{G} \gets \mathbf{G} \cdot G$
        \If{$m \cup_{|\mathbf{G}|} G$ is valid}\label{line:try_extend_2}
        \Comment{try $\varphi$ on a new graph}
            \State $m \gets m \cup_{|\mathbf{G}|} \varphi$
                \State $\textsc{EnumeratePartial}(i+1)$
            \State remove $\varphi$ from $m$
        \EndIf
        \State remove G from $\mathbf{G}$
    \EndFor
 \EndProcedure
\end{algorithmic}
\end{algorithm}

%%%%%%%%%%%%%%%%%%%%%%%%%%%%%%%%%%%%%%%%%%%%%%%%%%%%%%%%%%%%
%%%%%%%%%%%%%%%%%%%%%%%%%%%%%%%%%%%%%%%%%%%%%%%%%%%%%%%%%%%%
%%%%%%%%%%%%%%%%%%%%%%%%%%%%%%%%%%%%%%%%%%%%%%%%%%%%%%%%%%%%
\section{Enumerating Non-isomorphic Derivations}\label{sec:non-iso-der}

Algorithm~\ref{alg:enumerate_derivations} introduced in Section~\ref{sec:match_enumeration} enumerates the possible direct derivations
one can create given a rule and a set of non-isomorphic graphs. Depending on
the problem domain, however, some of the enumerated derivations might be redundant. As an example, assume we are given a rule $p$, two direct
derivations $\Der{G_1}{m_1, p} {H_1}$, $\Der{G_2}{m_2, p} {H_2}$ enumerated by Algorithm \ref{alg:enumerate_derivations}, and
an isomorphism $\varphi$ between $G_1$ and $G_2$ such that $\varphi \circ m_1 =
m_2$. Then, clearly, the host graphs $H_1 \simeq H_2$ must also be isomorphic. Only one of the two derivations is therefore relevant in multiple application scenarios, such as generation of reaction networks~\cite{andersen2016software}.

To identify such redundant derivations, we utilize the concept of derivation isomorphism.
Such isomorphism necessarily depends on the building blocks of the derivation itself, that is the host graph, the rule and the match.
A simple graph isomorphism suffices for the host graph, we thus begin with isomorphism of the graph transformation rules.

\begin{definition}[Rule Isomorphism]
Let $p_1 = (L_1 \xleftarrow{l_1}{} K_1 \xrightarrow{r_1}{} R_1)$ and $p_2 =
(L_2 \xleftarrow{l_2}{} K_2 \xrightarrow{r_2}{} R_2)$ be two rules.
We say that a triplet of graph isomorphisms $\varphi= (\varphi_L\colon
L_1\rightarrow L_2, \varphi_K\colon K_1\rightarrow K_2, \varphi_R\colon R_1\rightarrow
R_2)$ is a rule isomorphism between $p_1$ and $p_2$ if the following diagram
commutes:
	\begin{center}
		\begin{tikzpicture}
			\node (L1) {$L_1$};
			\node[right=of L1] (K1) {$K_1$};
			\node[right=of K1]  (R1) {$R_1$};
			\node[below=of L1]  (L2) {$L_2$};
			\node[right=of L2]  (K2) {$K_2$};
			\node[right=of K2]  (R2) {$R_2$};

			\draw[eMorphism] (K1) to node[above] {$l_1$} (L1);
			\draw[eMorphism] (K1) to node[above] {$r_1$} (R1);

			\draw[eMorphism] (L1) to node[right] {$\varphi_L$} (L2);
			\draw[eMorphism] (K1) to node[right] {$\varphi_K$} (K2);
			\draw[eMorphism] (R1) to node[right] {$\varphi_R$} (R2);

			\draw[eMorphism] (K2) to node[below] {$l_2$} (L2);
			\draw[eMorphism] (K2) to node[below] {$r_2\vphantom{l}$} (R2);
		\end{tikzpicture}
	\end{center}

If $p_1 = p_2$ we say $\varphi$ is a rule automorphism of $p_1$.
The internal graph isomorphisms $\varphi_L$, $\varphi_K$, $\varphi_R$ of a rule automorphism $\varphi$ are necessarily also graph automorphisms.
\end{definition}

A rule automorphism perfectly preserves the transformation prescribed within
the rule itself. Thus, given a pair of valid matches $m$ and $m'$ such that $m'
= m \circ \alpha_L$ where $\alpha_L$ is the first component of a rule
automorphism $\alpha$, the derivations $\Der{G}{m, p} {H_1}$ and
$\Der{G}{m', p} {H_2}$ necessarily produce isomorphic results
$H_1\simeq H_2$. Here, $m$ and $m'$ are an example of what we refer to as
isomorphic matches. In general, isomorphic partial matches are defined as
follows.

\begin{definition}[Isomorphic Partial Matches]\label{def:iso-matches}
	Let $G_1$, $G_2$ be two graphs, $p = (L \xleftarrow{l}{} K	\xrightarrow{r}{} R)$ a rule,
	and $m_1\colon L \pmap G_1$, $m_2\colon L \pmap G_2$ two valid partial matches of $L$ into $G_1$ and $G_2$ respectively.

	We say the partial matches $m_1$ and $m_2$ are isomorphic, $m_1 \simeq m_2$, if there exists an
	automorphism $(\alpha_L, \alpha_K, \alpha_R)$ of the rule $p$ and a graph isomorphism $\varphi_G\colon G_1 \rightarrow G_2$ such that the following diagram commutes:

	\begin{center}
		\begin{tikzpicture}
			\node (L) {$L$};
			\node[left=of L] (L2) {$L$};

			\node[below=of L] (G) {$G_1$};
			\node[below=of L2] (G2) {$G_2$};

			\draw[eMorphism, -left to] (L) to node[right] {$m_1$} (G);
			\draw[eMorphism, -left to] (L2) to node[right] {$m_2$} (G2);

			\draw[eMorphism] (L) to node[above] {$\alpha_L$} (L2);
			\draw[eMorphism] (G) to node[below] {$\varphi_G$} (G2);
		\end{tikzpicture}
	\end{center}
\end{definition}

Building on rule isomorphisms, match isomorphisms also preserve the structure of the graph transformation, guaranteeing isomorphism of the target graphs.
Hence, we can extend this notion of structure preservation to direct derivations.

\begin{definition}[Isomorphic Derivations]
	Let $p = (L \xleftarrow{l}{} K \xrightarrow{r}{} R)$ be a rule, and
	$\Der{G_1}{p,m_1}{H_1}$ and $\Der{G_2}{p,m_2}{H_2}$ two proper derivations.
	Then the two derivations are isomorphic if $m_1 \simeq m_2$.
\end{definition}

By only considering non-isomorphic derivations Problem \ref{prob:derivations_for_graph_set} can be restated as:

\begin{problem}
	Let $\mathcal{G}$ be a finite set of pairwise non-isomorphic connected graphs and rule $p = \Rule{L}$ a rule.

	Then we want to enumerate all pairwise \emph{non-isomorphic} proper derivations
	\[
	    \Der{\widetilde{\mathbf{G}}}{p,m}{H}
	\]
	where $\mathbf{G} \in {\mathcal{G}}^n$ for any $n \in \mathbb{N}$.

	\label{prob:canon}
\end{problem}

Isomorphic matches, as per Definition~\ref{def:iso-matches}, defines an equivalence relation on the set of matches enumerated by Algorithm
\ref{alg:enumerate_derivations}.
The solution to Problem~\ref{prob:canon} thus entails enumerating exactly one representative match for each of the classes of said equivalence relation. Efficiently constructing such representatives
in an iterative manner, however, turns out to be nontrivial due to the complex interplay between rule and host graph automorphisms.

In the following, we propose a heuristic for pruning partial matches, which can
be determined to only extend into derivations isomorphic to ones already enumerated.
The heuristic aims to identify canonical representatives for each equivalence class based on a total order on matches.
The total order itself is composed of two subparts.
The first part imposes an order on the way we extend by monomorphisms of
isomorphic component graphs of the left side of the rule.
The second part prunes matches based on the automorphisms 
in a single connected component either of the left side of the rule or of the host graph.

The heuristic does not guarantee that we do not enumerate any isomorphic matches.
The details of such cases are described towards the end of this section.
Nonetheless, the results in Section~\ref{sec:experiments} show the heuristic provides a significant speedup to the algorithm when one is only interested in enumerating non-isomorphic derivations.

\subsection{Order-Preserving Match Extensions}
Consider a rule $p =\Rule{\widetilde{\mathbf{L}}}$ where $\mathbf{L} = (L_1, \dots L_k)$ and let $L_i, L_j \in \mathbf{L}$ such that $L_i\simeq L_j$.
Even though $L_i$ and $L_j$ are isomorphic, their semantics
might differ in $p$, i.e., the corresponding isomorphism does not exist in $R$. We say that
the graphs $L_i$ and $L_j$ are isomorphic with respect to $p$, denoted
$L_i\simeq_p L_j$, if there exists an automorphism
$(\alpha_L, \alpha_K, \alpha_R)$ of $p$ such that $\alpha_L(\widetilde{\mathbf{L}}_i) = \widetilde{\mathbf{L}}_j$.

Given such components $L_i$ and $L_j$ and a match $m$, we can always obtain a new match $m' \simeq m$ by swapping the monomorphisms used to map $L_i$ and $L_j$.
Formally, to be able to exchange the monomorphisms, we have to modulate them with the automorphism $\alpha_L$, respectively its inverse.
For simplicity of the notation, we instead assume, without loss of generality, that any components $L_i \mathbin{\simeq_p}L_j$ are representationally equivalent, $L_i=L_j$ for the remainder of the section.

The existence of two components $L_i \mathbin{\simeq_p}L_j$ automorphic with respect to the rule $p$ can thus force Algorithm~\ref{alg:enumerate_derivations} to enumerate isomorphic matches.
Several examples of such isomorphic matches are given in Figure~\ref{fig:order-example}.
Since the monomorphisms mapping the components $L_i$ and $L_j$ are exchangeable, it is the order of said monomorphisms that differentiates between the isomorphic matches.
In this section, we define a total order on the monomorphisms, which can be used to identify a canonical representative among the isomorphic matches produced by monomorphism swapping between rule automorphic components $L_i \mathbin{\simeq_p}L_j$.

\begin{figure}[tbp]
    \begin{subfigure}[b]{1\textwidth}
    \centering
    \begin{tikzpicture}[gCat]
        \node[vCat, label=above:$\tug{L}$] (L) {
            \begin{tikzpicture}[gDot]
              \node[vDot, label={above right:$(1,1)$}, label={below right:\phantom{$(1,1)$}}] (l1) {};
              \node[vDot2, label={above right:$(2, 1)$}, right=of l1] (l2) {};
              \node[vDot2, label={above right:$(3, 1)$}, right=of l2] (l3) {};
            \end{tikzpicture}
        };
        \node[vCat, label=above:$K$, right=of L] (K) {
            \begin{tikzpicture}
              \node[vDot] (l1) {};
              \node[vDot2, right=of l1] (l2) {};
              \node[vDot2, right=of l2] (l3) {};
            \end{tikzpicture}
            };
        \node[vCat, label=above:$R$, right=of K] (R) {
            \begin{tikzpicture}
              \node[vDot] (l1) {};
              \node[vDot2, right=of l1] (l2) {};
              \node[vDot2, right=of l2] (l3) {};
            \end{tikzpicture}
        };
        \path[eMorphism] (K) edge node [above] {$l$} (L);
        \draw[eMorphism] (K) edge node [above] {$r$} (R);
    \end{tikzpicture}%
    \caption{An identity rule with three connected components in $L$. Each component contains a single vertex,
    		either a filled \tikz\node[vDot] {}; type or a hollow \tikz\node[vDot2] {}; type.}\label{fig:order-rule}
    \end{subfigure}
    
    %%%%%%%%%%%%%%%%%%%%%%%%%%%%%%%%%%%%%%%%%%%%%%%%%%%%%%%%%
    %%%%%%%%%%%%%%%%%%%%%%%%%%%%%%%%%%%%%%%%%%%%%%%%%%%%%%%%%
    \begin{subfigure}[b]{1\textwidth}
    \centering
            \begin{tikzpicture}
              \node[vDot] (l1) {};
              \node[vDot2, right=of l1] (l2) {};
              \draw (l1) to node[below] {$g\vphantom{h}$} (l2);
            \end{tikzpicture}
            \quad
            \begin{tikzpicture}
              \node[vDot] (l1) {};
              \node[vDot2, right=of l1] (l2) {};
              \node[vDot2, right=of l2] (l3) {};
              \draw (l1) to (l2);
              \draw (l2) to (l3);
              \path (l1) to node[below] {$h\vphantom{g}$} (l3);
            \end{tikzpicture}
    \caption{Three graphs used for the examples below.
    		As example we consider the following total order for these graphs:
    		$g\prec h$.}\label{fig:order-graphs}
    \end{subfigure}
    
    %%%%%%%%%%%%%%%%%%%%%%%%%%%%%%%%%%%%%%%%%%%%%%%%%%%%%%%%%
    %%%%%%%%%%%%%%%%%%%%%%%%%%%%%%%%%%%%%%%%%%%%%%%%%%%%%%%%%
	\tikzset{eMorphismGood/.style={eMorphism, dashed, blue}}
	\tikzset{eMorphismBad/.style={eMorphism, dashed, red}}
	\tikzset{eGraph/.style={transform canvas={yshift=-0.5mm}}}
	\begin{subfigure}[b]{1\textwidth}
	\begin{minipage}{0.45\textwidth}
    \caption{Two matches onto the graph of $\mathbf{G} = (g, h, g)$.
    		The red match it is mapped to a copy of $h$ whereas $L_3$ is mapped to a copy of $g$.
    		As $g\prec h$ this red match violates Condition \ref{itm:graph_order} of Definition~\ref{def:monomorphism_order}.
    		\label{fig:order-1}}
	\end{minipage}\hfill
	\begin{minipage}{0.48\textwidth}
	\centering
	\begin{tikzpicture}[gCat, remember picture]
	\node[vCat, label=above:$\tug{L}$] (L) {
		\begin{tikzpicture}[gDot]
		\node[vDot, label={above right:$(1,1)$}, label={below right:\phantom{$(1,1)$}}] (l1) {};
		\node[vDot2, label={above right:$(2, 1)$}, right=of l1] (l2) {};
		\node[vDot2, label={above right:$(3, 1)$}, right=of l2] (l3) {};
		\end{tikzpicture}
	};
	\node[vCat, label=below:$\tug{G}$, below=of L] (G) {
		\begin{tikzpicture}[gDot]
		\node[vDot] (g1-1) {};
		\node[vDot2, right=of g1-1] (g1-2) {};
		\draw (g1-1) to node[below, minimum size=0] {\footnotesize $G_1 = g\vphantom{h}$} (g1-2);

         \node[vDot, right=1.5em of g1-2] (g2-1) {};
         \node[vDot2, right=of g2-1] (g2-2) {};
         \node[vDot2, right=of g2-2] (g2-3) {};
         \draw (g2-1) to (g2-2);
         \draw (g2-2) to (g2-3);
         \path (g2-1) to node[below, minimum size=0] {\footnotesize $G_2 = h\vphantom{h}$} (g2-3);
	              
		\node[vDot, right=1.5em of g2-3] (g3-1) {};
		\node[vDot2, right=of g3-1] (g3-2) {};
		\draw (g3-1) to node[below, minimum size=0] {\footnotesize $G_3=g\vphantom{g}$} (g3-2);
		\end{tikzpicture}
	};
	\foreach \l/\g/\b in {l1/g1-1/-15, l2/g3-2/0, l3/g2-2/0}
		\draw[eMorphismGood, bend left=\b] (\l) to (\g);
	\foreach \l/\g/\b in {l1/g1-1/15, l2/g2-2/0, l3/g3-2/0}
		\draw[eMorphismBad, bend left=\b] (\l) to (\g);
	\end{tikzpicture}%
	\end{minipage}
    \end{subfigure}
    %%%%%%%%%%%%%%%%%%%%%%%%%%%%%%%%%%%%%%%%%%%%%%%%%%%%%%%%%
    %%%%%%%%%%%%%%%%%%%%%%%%%%%%%%%%%%%%%%%%%%%%%%%%%%%%%%%%%
	\begin{subfigure}[b]{1\textwidth}
	\begin{minipage}{0.45\textwidth}
    \caption{Two matches onto the graph of $\mathbf{G} = (g, g)$.
	Both $L_2$ and $L_3$ are mapped to copies of $g$, but in the red match $L_2$ is mapped to a copy
	with a higher index in $\mathbf{G}$ than $L_3$ is, and the red match thus
	violates Condition \ref{itm:component_order} of Definition \ref{def:monomorphism_order}.
    \label{fig:order-2}}
	\end{minipage}\hfill
	\begin{minipage}{0.48\textwidth}
	\centering
	\begin{tikzpicture}[gCat, remember picture]
	\node[vCat, label=above:$\tug{L}$] (L) {
		\begin{tikzpicture}[gDot]
		\node[vDot, label={above right:$(1,1)$}, label={below right:\phantom{$(1,1)$}}] (l1) {};
		\node[vDot2, label={above right:$(2, 1)$}, right=of l1] (l2) {};
		\node[vDot2, label={above right:$(3, 1)$}, right=of l2] (l3) {};
		\end{tikzpicture}
	};
	\node[vCat, label=below:$\tug{G}$, below=of L] (G) {
		\begin{tikzpicture}[gDot]
		\node[vDot] (g1-1) {};
		\node[vDot2, right=of g1-1] (g1-2) {};
		\draw (g1-1) to node[below, minimum size=0] {\footnotesize $G_1 = g\vphantom{h}$} (g1-2);
	              
		\node[vDot, right=of g1-2] (g2-1) {};
		\node[vDot2, right=of g2-1] (g2-2) {};
		\draw (g2-1) to node[below, minimum size=0] {\footnotesize $G_2 = g\vphantom{h}$} (g2-2);
		\end{tikzpicture}
	};
	\foreach \l/\g/\b in {l1/g1-1/-15, l2/g1-2/0, l3/g2-2/0}
		\draw[eMorphismGood, bend left=\b] (\l) to (\g);
	\foreach \l/\g/\b in {l1/g1-1/15, l2/g2-2/0, l3/g1-2/0}
		\draw[eMorphismBad, bend left=\b] (\l) to (\g);
	\end{tikzpicture}
	\end{minipage}
    \end{subfigure}
    %%%%%%%%%%%%%%%%%%%%%%%%%%%%%%%%%%%%%%%%%%%%%%%%%%%%%%%%%
    %%%%%%%%%%%%%%%%%%%%%%%%%%%%%%%%%%%%%%%%%%%%%%%%%%%%%%%%%
	\begin{subfigure}[b]{1\textwidth}
	\begin{minipage}{0.45\textwidth}
    \caption{Two matches onto the graph of $\mathbf{G} = (g, x)$.
	The graphs $L_2$ and $L_3$ are mapped to the same graph in $\mathbf{G}$,
	but in the red match the vertex of $L_3$ is mapped to a vertex with a lower index in $x$
	than the vertex of $L_2$ is, and the red match thus
	violates Condition~\ref{itm:match_order} of Definition \ref{def:monomorphism_order}.
    \label{fig:order-3}}
	\end{minipage}\hfill
	\begin{minipage}{0.48\textwidth}
	\centering
	\begin{tikzpicture}[gCat, remember picture]
	\node[vCat, label=above:$\tug{L}$] (L) {
		\begin{tikzpicture}[gDot]
		\node[vDot, label={above right:$(1,1)$}, label={below right:\phantom{$(1,1)$}}] (l1) {};
		\node[vDot2, label={above right:$(2, 1)$}, right=of l1] (l2) {};
		\node[vDot2, label={above right:$(3, 1)$}, right=of l2] (l3) {};
		\end{tikzpicture}
	};
	\node[vCat, label=below:$\tug{G}$, below=of L] (G) {
		\begin{tikzpicture}[gDot]
		\node[vDot] (g1-1) {};
		\node[vDot2, right=of g1-1] (g1-2) {};
		\draw (g1-1) to node[below, minimum size=0] {\footnotesize $G_1 = g\vphantom{h}$} (g1-2);
	              
		\node[vDot, label={below left:$(1)$}, right=of g1-2] (g2-1) {};
		\node[vDot2, label={above right:$(2)$}, right=of g2-1] (g2-2) {};
		\node[vDot2, label={below right:$(3)$}, right=of g2-2] (g2-3) {};
		\draw (g2-1) to (g2-2);
		\draw (g2-2) to (g2-3);
		\path (g2-1) to node[below, minimum size=0] {\footnotesize $G_2 = x\vphantom{h}$} (g2-3);
		\end{tikzpicture}
	};
	\foreach \l/\g/\b in {l1/g1-1/-15, l2/g2-2/0, l3/g2-3/0}
		\draw[eMorphismGood, bend left=\b] (\l) to (\g);
	\foreach \l/\g/\b in {l1/g1-1/15, l2/g2-3/0, l3/g2-2/0}
		\draw[eMorphismBad, bend left=\b] (\l) to (\g);
	\end{tikzpicture}
	\end{minipage}
    \end{subfigure}
    \caption{Examples of order-preserving (blue) and non-order-preserving (red)
    matches for the rule \subref{fig:order-rule} and graphs in \subref{fig:order-graphs}.
    The connected components $L_2$ and $L_3$ of $\tug{L}$ are isomorphic with respect to the rule $p$,
    and each of Figure \subref{fig:order-1}, \subref{fig:order-2}, and \subref{fig:order-3} show different
    examples of how a match can be non-order-preserving, corresponding to the conditions of Definition \ref{def:monomorphism_order}.}
    \label{fig:order-example}
\end{figure}
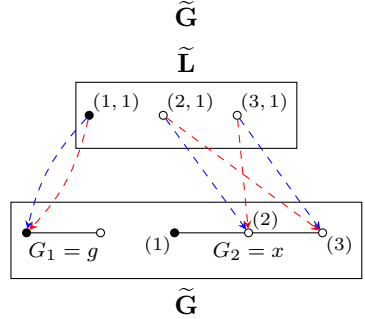

We construct the order in three steps.
First, we compare on the host graphs of the monomorphisms.
In case of equality, we compare on the indices of the target host graph components within the union graph.
And finally, if both monomorphisms map into the same component graph, we compare on the monomorphisms themselves, represented as vectors.

The first step requires us to fix a total order on the set of graphs $\mathcal{G}$. The set being finite means we can always find such an order, let thus ${\preceq} \subseteq \mathcal{G}\times \mathcal{G}$ be an arbitrary total order.
We can proceed with the definition of \emph{monomorphism order}.

\begin{definition}[Monomorphism Order]
	Let $p = \Rule{\widetilde{\mathbf{L}}}$ be a rule with $\mathbf{L} = (L_1, \dots, L_k)$ and $\mathbf{G} = (G_1, \dots, G_n)$ a union graph.
	Let further $\varphi\colon L_i \rightarrow G_x$, $\varphi'\colon L_j \rightarrow G_y$ be two monomorphisms mapping
	two isomorphic rule left side component graphs $L_i\mathbin{\simeq_p} L_j$
	into component graphs of $\mathbf{G}$.
	We say $\varphi < \varphi'$ if they satisfy one of the following conditions:
	\begin{enumerate}[label=(\arabic*)]
		\item $G_x \prec G_y$; \label{itm:graph_order}
		\item $G_x = G_y$ and $x < y$; \label{itm:component_order}
		\item $G_x = G_y$, $x = y$ and there exists $v \in \{1, \dots, |V(L_i)|\}$, $\varphi(v) < \varphi'(v)$ and for any $u < v$, $\varphi(u) = \varphi'(u)$, under the assumption of representation equivalence, $L_i =L_j$. \label{itm:match_order}
	\end{enumerate}

	\label{def:monomorphism_order}
\end{definition}

The monomorphism order defines a total order on monomorphisms into components of a given union graph.
The order is determined by the host component graphs of the monomorphisms (Conditions~\ref{itm:graph_order} and~\ref{itm:component_order}),
or as a lexicographical vector order on the images of their ordered domains in case the monomorphisms share a common host component graph (Condition~\ref{itm:match_order}).

The monomorphism order naturally extends into a total order $\ll$ on valid partial matches $m,m'\colon \mathbf{L} \pmap \mathbf{G}$ sharing a common domain,
$\dom(m) = \dom(m')$.
We say the match $m$ is smaller than $m'$, $m \ll m'$ if the vector of monomorphisms for each component of $\mathbf{L}$ is lexicographically smaller than the same vector for $m'$, formally $m \ll m'$ if there exists $i\in \{1,\dots, |\mathbf{L}|\}$, such that $m_i < m'_i$ and $m_j = m'_j$ for all $j < i$, .

We say that a match extension is \emph{order-preserving} if the monomorphism mapping the new component is larger than the monomorphisms mapping any preceding rule automorphic rule left side components.

\begin{definition}[Order-Preserving Match Extension]
    Let $m\colon \widetilde{\mathbf{L}} \pmap \widetilde{\mathbf{G}}$ be a valid partial match of a rule $p = \Rule{\widetilde{\mathbf{L}}}$ with $\mathbf{L} = (L_1, \dots, L_k)$
    into the graph $\widetilde{\mathbf{G}}$ with $\mathbf{G} = (G_1, \dots G_n)$.
    Let further $\varphi\colon L_i \rightarrow G_y$ be a monomorphism for some $i \in \{1, \dots, k\}$, $y \in \{1, \dots, n\}$ such that the extension $m \mathbin{\cup^i_y} \varphi$ is valid.

	We say that the extension $m \mathbin{\cup^i_y} \varphi$ is order-preserving, if for each $j \in \{1, \dots, i-1\}$ with $L_j \mathbin{\simeq_p} L_i$, $m_j < \varphi$.

    \label{def:order_extension}
\end{definition}

In simple terms, a partial match extension is order-preserving if the order on the matches of rule isomorphic components $L_i \mathbin{\simeq_p} L_j$ of the rule left side graph coincides with their order in $\mathbf{L}$.

Note that Definition~\ref{def:order_extension} only checks monomorphism order on component graphs $L_j \mathbin{\simeq_p} L_i$ in the prefix of the match, $j < i$.
This is in line with enumeration of partial matches by Algorithm~\ref{alg:enumerate_derivations}, which already follows the order of the component graphs of $\mathbf{L}$.
All partial matches enumerated by Algorithm~\ref{alg:enumerate_derivations} are thus defined on a prefix of $\mathbf{L}$.
This constraint is reflected in the definition of an \emph{order-preserving partial match}.

\begin{definition}[Order-Preserving Partial Match]
    Let $m\colon \widetilde{\mathbf{L}} \pmap \widetilde{\mathbf{G}}$ be a valid partial match of a rule $p = \Rule{\widetilde{\mathbf{L}}}$ with $\mathbf{L} = (L_1, \dots, L_k)$
    into the graph $\widetilde{\mathbf{G}}$ such that $m$ is defined exactly on $L_1, \dots, L_i$ for some $i \in \{0, \dots, k\}$

    Then $m$ is order-preserving if it can be decomposed into a series of order-preserving extensions $m = \varphi_1 \cup \dots \cup \varphi_i$.

	\label{def:order_match}
\end{definition}

Examples of order-preserving and non-order-preserving matches are depicted in
Figure \ref{fig:order-example}.

\begin{lem}
    Let $m\colon \widetilde{\mathbf{L}} \pmap \widetilde{\mathbf{G}}$ be a valid partial match of a rule $p = \Rule{\widetilde{\mathbf{L}}}$ with $\mathbf{L} = (L_1, \dots, L_k)$
    into the graph $\widetilde{\mathbf{G}}$ with $\mathbf{G} = (G_1, \dots G_n)$.

    Then there exists an order-preserving valid partial match $m'\colon \widetilde{\mathbf{L}} \pmap \widetilde{\mathbf{G}}$ such that $m \simeq m'$.

    \label{lem:order_match}
\end{lem}

\begin{proof}
Assuming $m$ is not already order-preserving, there must exist two indices $j<i \in \{1, \dots, k\}$ such that $L_i \mathbin{\simeq_p} L_j$ and $m_j > m_i$.
By definition, $L_i \mathbin{\simeq_p} L_j$ guarantees that there exists a rule automorphism $\alpha_p = (\alpha_L, \alpha_K, \alpha_R)$
such that $\alpha_L(L_i) = L_j$, $\alpha_L(L_j) = L_i$ and $\alpha_L$ limited to any other component of $\mathbf{L}$ is an identity.
We can use the automorphism to obtain the partial match $m' = m \circ \alpha_L$ isomorphic to $m$.
As $\alpha_L$ effectively swaps the monomorphisms mapping $L_i$ and $L_j$ we know $m'_j = m_i < m_j$.
For any index $l$ other than $i$ and $j$ we have $m_l=m'_l$. In particular this
holds for any $l < j$ and thus $m' \ll m$.

If $m'$ is order-preserving, we are done.
Otherwise, we repeat the procedure for $m'$, the termination guaranteed by $\ll$ and the number of valid partial matches of $\widetilde{\mathbf{L}}$ into $\widetilde{\mathbf{G}}$ being finite.
\end{proof}

Following Lemma~\ref{lem:order_match}, it suffices to enumerate only order-preserving matches to solve Problem~\ref{prob:canon}.
Assuming that rule left side component graphs $L_i \mathbin{{\simeq}_p} L_j$ related by rule automorphism are established in a preprocessing step, checking for order-preserving partial matches reduces to a simple check of the three conditions given in Definition~\ref{def:monomorphism_order}.

\subsection{Minimal Match Extensions}

In the previous section, we discussed how order-preserving matches allow us to
deal with isomorphisms across multiple rule left side graph components. Recall
that by Property \ref{prop:order_independence}, the way that Algorithm
\ref{alg:enumerate_derivations} enumerates derivations we are ensured that
matches are enumerated only up to the reordering of the host graph. As a
result, by only enumerating order-preserving partial matches, we are enumerating 
derivations up to the reordering of the host graph as well as the reordering of isomorphic 
connected components of the left graph of a rule.

Order-preserving matches do not, however, consider isomorphic matches
obtained from automorphisms that permutes single component graphs of rule nor
host graphs.

As an example, assume we are given a valid match $m\colon \tug{L} \rightarrow
\tug{G}$, where $\ug{L} = (L_1)$ and $\ug{G} = (G_1)$, and an automorphism
$\alpha_{G_1}\colon G_1 \rightarrow G_1$. As $m$ consists of
a single connected component, clearly it is also order-preserving. However,
the match $m' = \alpha_{G_1} \circ m$ is isomorphic to $m$.

Given a rule automorphism $(\alpha_L, \alpha_K, \alpha_R)$, we let $\alpha_{L,i}
\subseteq \alpha_L$ refer to all the automorphisms of $\alpha_L$ that permutes
vertices of the $i$th connected component of $L$.

We refer to such automorphisms $\varphi_{L,i}$ as \emph{local rule automorphisms}.
For a local rule automorphism $\varphi_{L,i}$, the match $m' = m \circ
\varphi_{L,i}$ is by definition valid and isomorphic to $m$, $m \simeq m'$.

The automorphisms of the component graphs can themselves be exploited to avoid enumeration of such isomorphic (partial) matches.
As we are dealing with automorphisms of single component graphs, we can limit
ourselves to the context of the monomorphisms $m_i\colon L_i \rightarrow G_x$
mapping the individual component graphs of $\mathbf{L}$.
We let two monomorphisms be isomorphic if their obvious corresponding partial
matches are isomorphic, and
re-utilize the monomorphism order from Definition~\ref{def:monomorphism_order}.
As relevant monomorphisms will always map into the same component graph of the host graph, only Condition~\ref{itm:match_order} of Definition~\ref{def:monomorphism_order} is applicable.
This order corresponds to the lexicographic order on the vectors obtained by
applying each monomorphism on their ordered domain.
A total order on monomorphisms uniquely identifies the minimal and maximal elements for each subset, in particular all sets of pairwise isomorphic monomorphisms.
As we only need one, we choose the minimal element as the unique representative.

\begin{definition}[Minimal Monomorphism]
    Let $p = \Rule{\widetilde{\mathbf{L}}}$ be a rule with $\mathbf{L} = (L_1, \dots, L_k)$, $i \in \{1, \dots, k\}$ an index and $G \in \mathcal{G}$ a graph. Let further $\mathcal{A}_{L,i}$ be a set of local rule automorphisms on $L_i$, and $\mathcal{A}_G$ a set of graph automorphisms of $G$.

    Then we say a monomorphism $\varphi\colon L_i \rightarrow G$ is minimal with respect to $\mathcal{A}_{L,i}$ and $\mathcal{A}_G$ if for each $\varphi' \in \{\alpha_{G} \circ \varphi \circ \alpha_{L,i} \mid \alpha_{L,i} \in \mathcal{A}_{L,i} \text{ and } \alpha_G \in \mathcal{A}_G\}$, $\varphi \leq \varphi'$.

    \label{def:minimal_morphism}
\end{definition}

Definition~\ref{def:minimal_morphism} identifies a unique representative for any set of pairwise isomorphic monomorphisms delimited by the local automorphism set $\mathcal{A}_{L,i}$ and graph automorphism set $\mathcal{A}_G$.
In the general case, when all relevant automorphisms are included, the obtained minimal monomorphisms are pairwise non-isomorphic.
While such a configuration might appear to be the obvious choice, in the context of a match, several rule left side component graphs might be mapped into the same host component graph.
Consider two indices $i<j \in \{1,\dots, k\}$, with $m_i\colon L_i \rightarrow G_x$ for some $x\in \{1, \dots, n\}$ and a monomorphism $\varphi\colon L_j \rightarrow G_x$ into the same host component graph.
Then, for the extension $m \cup \varphi$, we are only interested in such isomorphisms of the monomorphism $\varphi$ which do not impact $\codom(m)$.
This can be achieved by taking $\mathcal{A}_{G,x}$ to be the set of all
automorphisms $\alpha_{G,x}\colon G_x\rightarrow G_x$ such that for each union graph vertex $(x,
v) \in \codom(m)$, $\alpha_{G,x}(v) = v$.
We denote such sets of automorphisms preserving $m$ as $\mathcal{A}_{G,x}[m] = \{\alpha_{G,x}\colon G_x\rightarrow G_x\mid \alpha_{G,x} \text{ is an automorphism and for all } (x, v) \in \codom(m),\,\alpha_{G,x}(v) = v\}$.
No restriction on the local rule automorphisms being necessary.

\begin{definition}[Minimal Match Extension]
    Let $m\colon \widetilde{\mathbf{L}} \pmap \widetilde{\mathbf{G}}$ be a valid partial match of a rule $p = \Rule{\widetilde{\mathbf{L}}}$ with $\mathbf{L} = (L_1, \dots, L_k)$ into a graph $\widetilde{\mathbf{G}}$ with $\mathbf{G} = (G_1, \dots, G_n)$.
    Let further $\varphi\colon L_i \rightarrow G_x$ be a monomorphism for some $i \in \{1, \dots, k\}$, $x \in \{1, \dots, n\}$ such that the extension $m \mathbin{\cup_x^i} \varphi$ is valid.

    We say that the extension $m \mathbin{\cup_x^i} \varphi$ is minimal, if $\varphi$ is minimal with respect to all local rule automorphisms of $L_i$ and $\mathcal{A}_{G,x}[m]$.

    \label{def:minimal_extension}
\end{definition}

As an extension of Definition~\ref{def:minimal_extension}, we say a valid partial match $m\colon \widetilde{\mathbf{L}} \pmap \widetilde{\mathbf{G}}$ defined on $L_1, \dots, L_j$ for some $j \in \{0, \dots, k\}$, is \emph{minimal} if it can be decomposed into a series of extensions $m = \varphi_1 \cup \dots \cup \varphi_j$ such that for each $1 \leq i \leq j$, the extension $(\varphi_1 \cup \dots \cup \varphi_{i-1}) \cup \varphi_i$ is minimal.

\begin{lem}
    Let $m\colon \widetilde{\mathbf{L}} \pmap \widetilde{\mathbf{G}}$ be a valid partial match of a rule $p = \Rule{\widetilde{\mathbf{L}}}$ with $\mathbf{L} = (L_1, \dots, L_k)$
    into the graph $\widetilde{\mathbf{G}}$ with $\mathbf{G} = (G_1, \dots, G_n)$.

    Then there exists a minimal partial match $m'\colon \widetilde{\mathbf{L}} \pmap \widetilde{\mathbf{G}}$ such that $m \simeq m'$.

    \label{lem:minimal_match}
\end{lem}

\begin{proof}
Assuming $m$ is not already minimal, there must exist an index $i \in \{1, \dots, k\}$ such that the extension $(m_1 \cup \dots \cup m_{i-1}) \cup m_i$ is not minimal.
By definition, there exist a rule automorphism $(\alpha_L, \alpha_K, \alpha_R)$ of $p$ such that $\alpha_L \in \mathcal{A}_{L,i}$ is an automorphism of $L_i$, and a host graph automorphism $\alpha_G \in \mathcal{A}_{G,x}[m_1 \cup \dots \cup m_{i-1}]$, which allow us to construct $\varphi = \alpha_G \circ m_i \circ \alpha_L$ such that $\varphi < m_i$.
We can use the automorphisms to obtain the partial match $m' = \alpha_G \circ m \circ \alpha_L$ isomorphic to $m$.
$\alpha_L$ only interacts with the rule left side graph component $L_i$ and $\alpha_G \in \mathcal{A}_{G,x}[m_1 \cup \dots \cup m_{i-1}]$ guarantees that for all $j < i$, $m'_j = m_j$.
Moreover, $\varphi = m'_i < m_i$ and thus $m' \ll m$.

If $m'$ is minimal, we are done. Otherwise, we repeat the procedure for $m'$, the termination guaranteed by $\ll$ and the number of valid partial matches of $\widetilde{\mathbf{L}}$ into $\widetilde{\mathbf{G}}$ being finite.
\end{proof}

Lemma~\ref{lem:minimal_match} guarantees that for any enumerated valid match we also enumerate at least one isomorphic match which is also minimal.
The solution to Problem~\ref{prob:canon} can thus be obtained by only enumerating minimal matches.
To be able to use Lemma~\ref{lem:minimal_match} in conjunction with Lemma~\ref{lem:order_match} stating that order-preserving partial matches are sufficient, we still have to prove that for any valid partial match there exists an isomorphic match that is both order-preserving and minimal.

\begin{lem}
	Let $m\colon \widetilde{\mathbf{L}} \pmap \widetilde{\mathbf{G}}$ be a valid partial match of a rule $p = \Rule{\widetilde{\mathbf{L}}}$ with $\mathbf{L} = (L_1, \dots, L_k)$ into the graph $\widetilde{\mathbf{G}}$ with $\mathbf{G} = (G_1, \dots, G_n)$.
	Then there exists a minimal and order-preserving partial match $m'\colon \widetilde{\mathbf{L}} \pmap \widetilde{\mathbf{G}}$ such that $m \simeq m'$.

\label{lem:minimal_order_match}
\end{lem}

\begin{proof}
	Let us first assume $m$ is not order-preserving.
	By Lemma~\ref{lem:order_match}, there exists an isomorphic order-preserving match $m' \simeq m$ which is additionally smaller than $m$ according to the lexicographic order on monomorphism vectors, $m' \ll m$.
	Similarly, if $m$ is not minimal, by Lemma~\ref{lem:minimal_match}, there exists an isomorphic minimal match $m'' \simeq m$ which is also smaller, $m'' \ll m$.
	
	As $\ll$ is a total order and there are finitely many partial matches from $\widetilde{\mathbf{L}}$ into $\widetilde{\mathbf{G}}$, the consecutive application of Lemma~\ref{lem:order_match} to obtain an order-preserving match and Lemma~\ref{lem:minimal_match} to obtain a minimal match must necessarily terminate in a fixed-point which is both order-preserving and minimal.
\end{proof}

Lemma~\ref{lem:minimal_order_match} allows us to obtain a solution to Problem~\ref{prob:canon} by only enumerating order-preserving minimal partial matches.
It should be noted, however, that two order-preserving minimal matches can still be isomorphic, as illustrated in Figure \ref{fig:invalid_iso_match}, making the method introduced in this section a heuristic.
This is due to the limitation to host component graph automorphisms which fix the partial match $m$, $\mathcal{A}_{G,x}[m]$, in Definition~\ref{def:minimal_extension}.

\begin{figure}[tbp]
    \centering
    \begin{subfigure}[b]{0.48\textwidth}
    \centering
    \begin{tikzpicture}[gCat, remember picture]
        \node[vCat, label=above:$\tug{L}$] (L) {
            \begin{tikzpicture}[gDot, remember picture]
              \node[vDot2, label={above left:$(1, 1)$}] (l1) {};
              \node[vDot2, label={above:$(2, 1)$}, right=of l1] (l2) {};
              \node[vDot2, label={above right:$(3, 1)$}, right=of l2] (l3) {};
            \end{tikzpicture}
        };
        \node[vCat, label=below:$\tug{G}$, below=of L] (G) {
            \begin{tikzpicture}[gDot, remember picture]
              \node[vDot2, label={below left:$1$}] (g1) {};
              \node[vDot2, label={below right:$2$}, right=of g1] (g2) {};
              \node[vDot2, label={below left:$3$}, below=of g1] (g3) {};
              \node[vDot2, label={below right:$4$}, right=of g3] (g4) {};
			 \draw (g1) to (g2);
			 \draw (g3) to (g4);
			 \draw (g1) to (g3);
			 \draw (g2) to (g4);
            \end{tikzpicture}
        };

        \draw[eMorphism, red, dashed, bend left] (l1) to (g1);
        \draw[eMorphism, red, dashed, bend left] (l2) to (g2);
        \draw[eMorphism, blue, dashed, bend right] (l1) to (g1);
        \draw[eMorphism, blue, dashed, bend right] (l2) to (g2);

		\draw[eMorphism, red, dashed, bend left] (l3)  to (g4);
        \draw[eMorphism, blue, dashed, bend right] (l3) to (g3);
    \end{tikzpicture}%
    \caption{
    		Example with no edges in $\tug{L}$, and all vertices having the same type.
    		The three connected components of $\tug{L}$ are pairwise isomorphic.
    }\label{fig:invalid_iso_match:no_edge}
    \end{subfigure}\hfill
    \begin{subfigure}[b]{0.48\textwidth}
    \centering
    \begin{tikzpicture}[gCat, remember picture]
        \node[vCat, label=above:$\tug{L}$] (L) {
            \begin{tikzpicture}[gDot, remember picture]
              \node[vDot2, label={above left:$(1, 1)$}] (l1) {};
              \node[vDot2, label={above:$(1, 2)$}, right=of l1] (l2) {};
              \node[vDot, label={above right:$(2, 1)$}, right=of l2] (l3) {};
		  \draw (l1) to (l2);
            \end{tikzpicture}
        };
        \node[vCat, label=below:$\tug{G}$, below=of L] (G) {
            \begin{tikzpicture}[gDot, remember picture]
              \node[vDot2, label={below left:$1$}] (g1) {};
              \node[vDot2, label={below right:$2$}, right=of g1] (g2) {};
              \node[vDot, label={below left:$3$}, below=of g1] (g3) {};
              \node[vDot, label={below right:$4$}, right=of g3] (g4) {};
			 \draw (g1) to (g2);
			 \draw (g3) to (g4);
			 \draw (g1) to (g3);
			 \draw (g2) to (g4);
            \end{tikzpicture}
        };

        \draw[eMorphism, red, dashed, bend left] (l1) to (g1);
        \draw[eMorphism, red, dashed, bend left] (l2) to (g2);
        \draw[eMorphism, blue, dashed, bend right] (l1) to (g1);
        \draw[eMorphism, blue, dashed, bend right] (l2) to (g2);

		\draw[eMorphism, red, dashed, bend left] (l3)  to (g4);
        \draw[eMorphism, blue, dashed, bend right] (l3) to (g3);
    \end{tikzpicture}%
    \caption{
    		Example with a single edge in $\tug{L}$, and two different vertex types.
    		The connected components of $\tug{L}$ are not isomorphic, but $L_1$ has a symmetry.
	}\label{fig:invalid_iso_match:edge}
    \end{subfigure}%
	\caption{Examples where two isomorphic matches are both minimal and order-preserving.
		The vertices of $\tug{L}$ are annotated with graph index and vertex index,
		while the vertices of $\tug{G}$ are annotated with just their vertex index.
		In each example one match is $m$ (blue) and the other $m'$ (red).
		In \subref{fig:invalid_iso_match:no_edge},
		let $\varphi_1$ and $\varphi_2$ represent the corresponding monomorphisms mapping $L_1$ and $L_2$ into $G_1$,
		and let $\varphi_3$ (resp.\ $\varphi'_3$) be the constituent monomorphism for $m$ (resp.\ $m'$) mapping $L_3$ into $G_1$.
		Then, we can write each match as a sequence of extensions, $m = \varphi_1 \cup \varphi_2 \cup \varphi_3$
		and $m' = \varphi_1 \cup \varphi_2 \cup \varphi'_3$.
		Clearly, each extension is order-preserving and hence both matches are order-preserving.
		To see why both are minimal matches, observe that $G_1$ contains the local automorphism
		$\alpha_{G_1} = \{1\mapsto 2, 2\mapsto 1, 3\mapsto 4, 4\mapsto 3\}$.
		However, $\alpha_{G_1}\not\in \mathcal{A}_{G,1}[\varphi_1\cup\varphi_2]$,
		and both $\varphi_3$ and $\varphi'_3$ are minimal match extensions.
		Only when the automorphism $\alpha_{G_1}$ is used in conjunction with the automorphism that swaps $L_1$ and $L_2$
		one can detect that $m$ and $m'$ are in fact isomorphic.
		In \subref{fig:invalid_iso_match:edge} the same issue occurs, though here the symmetry of $\tug{L}$ is within a single connected component.
	}
	\label{fig:invalid_iso_match}
\end{figure}
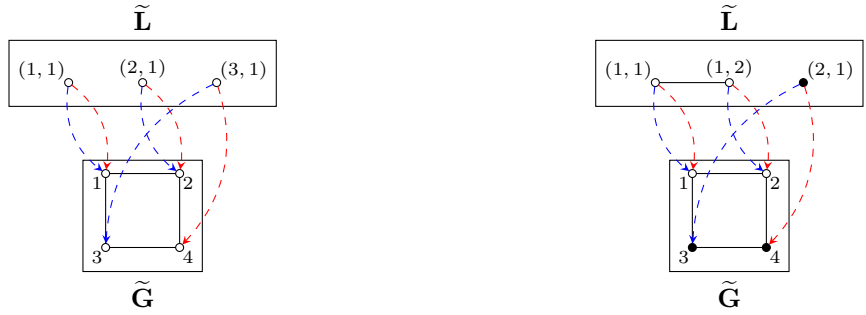

In practice, more automorphisms could be considered to obtain less minimal matches, if they commute with a rule automorphism on the match $m$.
In particular, consider a graph automorphism $\alpha_{G,x}\colon G_x
\rightarrow G_x$ which does not satisfy the condition $\alpha_{G,x}(v) = v$
for each union graph vertex $(x, v) \in \codom(m)$, but for which there exists a rule automorphism with
$\alpha_L\colon \widetilde{\mathbf{L}} \rightarrow \widetilde{\mathbf{L}}$
such that for each $u \in \dom(m)$ such that $m(u) = (x, v)$ for some $v \in
V(G_x)$, $m(u) = \alpha_{G,x}\circ m\circ \alpha_L(u)$ as shown in Figure
\ref{fig:invalid_iso_match}.
The left side of the commuting rule automorphism $\alpha_L$ allows us to ``reverse'' the change
introduced to $m$ by $\alpha_{G,x}$ thus preserving the isomorphism of the
results, $m \mathbin{\cup_x} (\alpha_{G,x} \circ \varphi) = \alpha_{G_x}
\circ (m \mathbin{\cup_x} \varphi) \circ \alpha_p$.

Detection of automorphisms $\alpha_{G,x}\colon G_x \rightarrow G_x$ with a
commuting rule automorphism is nontrivial, as it in the general case requires
enumeration over all rule automorphisms, including non-local
ones.
In spirit of avoiding explicit enumeration of all rule automorphisms, whose
number can be up to exponentially larger than the number of local rule
automorphisms, we limit ourselves to automorphisms that fix the match $m$.
While an efficient discovery of automorphisms $\alpha_{G,x}$ for which there
exists a commuting rule automorphisms might be possible, as our results in
Section~\ref{sec:experiments} showcase, using a limited amount of symmetries is
sufficient to obtain significant speedup.
We thus consider such exploration to be beyond the scope of this contribution.

Finally, the implementation of the symmetry pruning can be done by a minor modification of Algorithm~\ref{alg:enumerate_derivations}.
In particular, it suffices to check the minimality and order-preservence of the extension on top of
the validity check on lines~\ref{line:try_extend_1} and~\ref{line:try_extend_2}.
Additionally, since the minimality of a monomorphism is always determined with
respect to all local rule automorphisms, we can pre-filter the
monomorphisms computed on line~\ref{line:store_mono} to only include ones that are minimal
with respect of all local rule automorphisms and the identity automorphism of
the host graph, $\mathcal{A}_G = \{\id_G\}$.

%%%%%%%%%%%%%%%%%%%%%%%%%%%%%%%%%%%%%%%%%%%%%%%%%%%%%%%%%%%%
%%%%%%%%%%%%%%%%%%%%%%%%%%%%%%%%%%%%%%%%%%%%%%%%%%%%%%%%%%%%
%%%%%%%%%%%%%%%%%%%%%%%%%%%%%%%%%%%%%%%%%%%%%%%%%%%%%%%%%%%%

\section{Implementation}
The algorithm described in the above sections, \emph{Efficient Derivation Enumeration} (EDE), is implemented in C++ as part of the software framework MØD \cite{andersen2016software}, available online: \url{https://github.com/jakobandersen/mod/tree/archive/rule-application-21}.
The MØD integration allows the algorithm to be easily employed for efficient reaction network construction.
Moreover, the algorithm relies on the solutions of several well-known problems, many already implemented in MØD.

First, \EDE{} constructs a database of monomorphisms between connected components in the rule and the given connected input graphs.
To enumerate these monomorphisms we employ the VF2 algorithm \cite{vf2} as implemented in the Boost Graph library \cite{boost}.
The reaction network is constructed by including all discovered, pairwise non-isomorphic derivations.
To check for isomorphisms between derivations, we use the graph canonicalization framework \cite{generic-graphcanon}
as implemented in the GraphCanon library \cite{gh:graphCanon}, allowing for efficient comparison between two derivations.
This framework is based on the individualization-refinement approach, used by many other the tools, e.g. nauty \cite{pgi:1}, Bliss \cite{bliss,bliss:2} or Traces \cite{pgi:2}.

The \EDE{} algorithm allows for multiple configurations.
We use \EDE{} for the simple version corresponding to Algorithm~\ref{alg:enumerate_derivations}.
The version with symmetry pruning, denoted \SEDE{} (\EDE{} with Symmetry pruning), only considers order-preserving and minimal extensions, as discussed in Section \ref{sec:non-iso-der}.
Conveniently, the individualization-refinement approach constructs the generators
of the automorphism group of the graph being canonicalized as a by-product, allowing \SEDE{} to use the generators in match minimality checking with no extra computation cost.
Similarly, canonicalization can be used to obtain automorphism group generators of rule left side connected components.

To construct the authomorphism group itself we employ the Schreier-Sims algorithm \cite{Sims1971}.
An implementation of this algorithm is provided by the GraphCanon library \cite{generic-graphcanon, gh:graphCanon}.
In practice, the construction of automorphism groups using the Schreier-Sims algorithm can be fairly
expensive with limited gain, especially in the case of graphs with few automorphisms.
The algorithm as presented in Section \ref{sec:non-iso-der}, however, does not require every automorphism of the group to be represented.
While using only a subset of automorphisms may lead to enumeration of some isomorphic derivations,
simple symmetries can be pruned based on the derived generators themselves.
Many of the symmetries commonly encountered in molecule graphs are indeed simple, e.g. hydrogen atoms of a carbon atom.
The abundance of simple symmetries often justifies skipping the costly automorphism group construction altogether,
motivating us to consider one more configuration, \SSEDE{} (\EDE{} with Simple Symmetry pruning), that only considers order-preserving and minimal matches, where minimality is determined using only the generators of the automorphism groups rather than the automorphism groups themselves.

\section{Experiments}
\label{sec:experiments}

Our work is heavily motivated by the need for efficient generation of chemical reaction networks.
We therefore employ reaction network generation as the testbed for comparing \EDE{} against the Partial Application (\PA{}) approach described in the introduction.
For the purposes of the experiments, we understand \PA{} in the limited scope of the rule application enumeration algorithm published in~\cite{andersen2016software}, which can be directly employed to solve Problem~\ref{prob:canon}.
Both \EDE{} and \PA{} are implemented within the MØD software framework, relying on the same infrastructure for the construction of the reaction network itself, thus ensuring a fair comparison.

We conduct four different experiments, two using artificial data and two using chemical data.
Each experiment consists of the iterative construction of a reaction network $N$
from an initial set of graphs $\mathcal{G}$ and a given set of graph transformation rules $\mathcal{R}$.
In each iteration, \EDE{} or \PA{} are used to enumerate all non-isomorphic derivations for each of the rules in $\mathcal{R}$ (Problem~\ref{prob:canon}) using the molecules of the network $N$ for the set $\mathcal{G}$.
The obtained derivations are then included in $N$, in the form of reactions linking the components of the host and result graphs, with any newly discovered graphs as new molecules included in $\mathcal{G}$ in the next iteration.

To demonstrate the differences in isomorphic match pruning performance, we employ \EDE{} in all three available settings (\EDE{}, \SEDE{} and \SSEDE{}).
These variations allow us to showcase the effect of early symmetry pruning on different examples.

Similarly, the existing \PA{} implementation also allows for different configurations with respect to isomorphism checking.
In particular, the implementation allows checking for isomorphisms between rules obtained by partial application.
Thus, if two different partial applications result in isomorphic rules, only one is preserved, effectively pruning isomorphic partial matches.
We consider \PA{} in both configurations, as plain \PA{} with no isomorphism checking (analogous to simple \EDE{}),
and as \IPA{} (\PA{} with Isomorphism checking) for the version where isomorphic rules are pruned after each partial application.

All results were obtained using a laptop computer equipped with a $9\textsuperscript{th}$ generation Intel core i7 processor and \SI{32}{\giga\byte} of RAM. The code has been compiled with the \texttt{g++} compiler version $8$, with optimization level \texttt{O3}.

\subsection{Binary Strings}
\label{sec:binary_strings}

The first experiment is designed to compare performance in a setting with no rule or host graph symmetries.
The experiment itself consists of connecting small molecules with a common ``backbone'' into chains, mimicking the arrangement of amino acids into polypeptide chains.
For the sake of simplicity,  the experiment follows a minimalistic design.
The common backbone of the molecules consists of a single C-O bond, leading to the chains being an interleaving of carbon and oxygen atoms.
The use of different atom types allows us to easily distinguish the ends of the chains (C-end and O-end) ensuring there are no symmetries in the constructed molecules.

We consider two initial graphs, shown in Figure~\ref{fig:binary_graph_a} and Figure~\ref{fig:binary_graph_b}, distinguished by their abstract side chain A or B.
From a formal language perspective, arranging these molecules into chains is equivalent to the construction of strings over the binary alphabet $\Sigma = \{\text{A}, \text{B}\}$.
We further consider the experiment in several variations based on the number~$k$ of connected components in the rule left side graph, for $k \in \{2, 3, 4\}$.
For each $k$ a different rule, $\texttt{chain}(k)$, is used to build chains by connecting $k$ molecules at a time.
An example for $k=3$ is shown in Figure~\ref{fig:binary_rule}.
For each $k$, the rule contains exactly one component that can match chains of more than one molecule.
This allows us to restrict the growth of the chains to only allow extensions of the O-end, and ensuring that there are no symmetries in the rules.
In terms of binary strings, each variation is building molecules corresponding to the regular expression $\{\text{A}, \text{B}\}^{n(k-1)+1}$ where $n \in \mathbb{N}$ is the number of rule applications.

To curtail the constructed reaction networks, which are theoretically infinite, we limit all variations to molecules of up to $22$ atoms, effectively restricting the binary strings to length $7$.
This restriction explains the drop in the size of the constructed network for increasing $k$ of $\texttt{chain(k)}$ (Figure~\ref{fig:binary_dg}) as it limits the number of network construction iterations.
Figure~\ref{fig:binary_calls} then depicts the total number of derivations reported by each algorithm (the number of yields in Algorithm~\ref{alg:enumerate_derivations}) across all iterations.
As there are no symmetries, all variations of \EDE{} compute the same number of derivations, corresponding to the size of the reaction network.
\PA{} on the other hand, computes multiple copies of the same derivation,
owing to the partial application not considering any order on the connected components of the rule left side graphs.
We can observe that the isomorphism checking for \IPA{} manages to eliminate some of the duplicate derivations---however, it does not achieve a higher efficiency as seen by the time taken in Figure~\ref{fig:binary_time}.
This follows from the number of partially applied rules growing exponentially with the number of connected components in the rule left side graph, making the pairwise isomorphism checking intractable.

Figure~\ref{fig:binary_time} further showcases that the speedup \EDE{} achieves over \PA{} seems consistent with the expected exponential in $k$.
Observe also that in spite of having no symmetries to exploit, \SEDE{} and \SSEDE{}
variations perform on par with \EDE{}.
This seems to indicate that the overhead introduced by order-preserving and minimality checks, including the construction of the symmetry groups, is negligible if no symmetries are present.

\begin{figure}
	\centering
	\begin{subfigure}[b]{0.25\textwidth}
	\centering
	\includegraphics[scale=0.7]{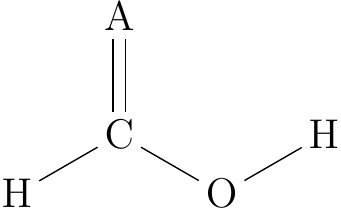}
	\caption{\label{fig:binary_graph_a}}
	\end{subfigure}
	\qquad
	\begin{subfigure}[b]{0.20\textwidth}
	\centering
	\includegraphics[scale=0.7]{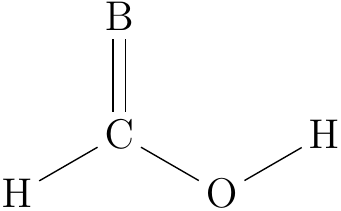}
	\caption{\label{fig:binary_graph_b}}
	\end{subfigure}

	\begin{subfigure}[b]{\textwidth}
	\centering
	\summaryRuleSpan[scale=0.7]{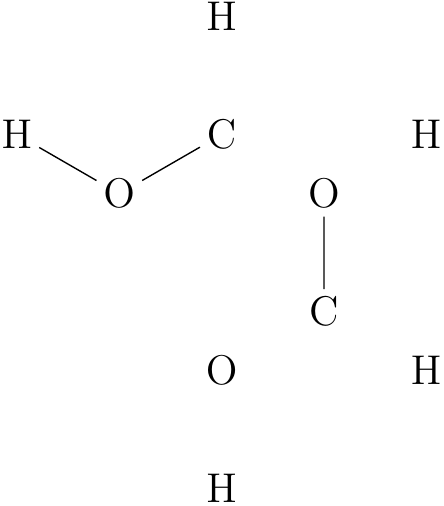}
	\caption{\label{fig:binary_rule}}
	\end{subfigure}
	\caption{Initial graphs and example rule for Binary Strings.
	\subref{fig:binary_graph_a}-\subref{fig:binary_graph_b} The initial graphs.
	\subref{fig:binary_rule} The rule $\texttt{Chain}(3)$.
	The two first connected components of $L$ (from left to right) are joined at the end of a chain modeled by the third, and final, connected component of $L$.
	\label{fig:binary_grammar}}
\end{figure}

\begin{figure}
	\centering
	\begin{subfigure}[b]{0.49\textwidth}
		\begin{tikzpicture}
			\begin{axis}[
			ybar,
			novo bar style,
			ylabel={elements [counts]},
			xlabel={k},
			xtick=data,
			xmin=1.5,
			xmax=4.5,
			legend cell align={left},
			]
			\addplot[fill=papergrey] table[x=name, y=|E(N)|]{figures/binary_strings/binary_strings_dg.dat};
			\addplot[fill=paperred] table[x=name, y=|V(N)|]{figures/binary_strings/binary_strings_dg.dat};
			\legend{$|E(N)|$,$|V(N)|$}
			\end{axis}
		\end{tikzpicture}
		\caption{}
		\label{fig:binary_dg}
		\end{subfigure}
		\hfill
		\begin{subfigure}[b]{0.49\textwidth}
			\begin{tikzpicture}
			\begin{axis}[
			ybar,
			novo bar style,
			ylabel={derivations [counts]},
			xlabel={k},
			xtick=data,
			bar width=3.5pt,
			xmin=1.5,
			xmax=4.5,
			yticklabel pos=right,
			ytick pos=right,
			legend pos=north west,
			legend cell align={left},
			]
			\addplot[fill=papergrey] table[x=name, y=PA]{figures/binary_strings/binary_strings_calls.dat};
			\addplot[fill=paperred] table[x=name, y=PA-I]{figures/binary_strings/binary_strings_calls.dat};
			\addplot[fill=papergreen] table[x=name, y=EDE]{figures/binary_strings/binary_strings_calls.dat};
			\addplot[fill=paperblue] table[x=name, y=EDE-S]{figures/binary_strings/binary_strings_calls.dat};
			\addplot[fill=paperorange] table[x=name, y=EDE-SS]{figures/binary_strings/binary_strings_calls.dat};
			\legend{PA, PA-I,EDE,EDE-S,EDE-SS}
			\end{axis}
			\end{tikzpicture}
			\caption{}
			\label{fig:binary_calls}
		\end{subfigure}
		\\[.5cm]
		\begin{subfigure}[b]{0.49\textwidth}
			\begin{tikzpicture}
			\begin{axis}[
			ybar,
			novo bar style,
			ylabel={time [\si{\second}]},
			xlabel={k},
			xtick=data,
			bar width=3.5pt,
			xmin=1.5,
			xmax=4.5,
			legend pos=north west,
			legend cell align={left},
			]
			\addplot[fill=papergrey] table[x=name, y=PA]{figures/binary_strings/binary_strings_time.dat};
			\addplot[fill=paperred] table[x=name, y=PA-I]{figures/binary_strings/binary_strings_time.dat};
			\addplot[fill=papergreen] table[x=name, y=EDE]{figures/binary_strings/binary_strings_time.dat};
			\addplot[fill=paperblue] table[x=name, y=EDE-S]{figures/binary_strings/binary_strings_time.dat};
			\addplot[fill=paperorange] table[x=name, y=EDE-SS]{figures/binary_strings/binary_strings_time.dat};
			\legend{PA, PA-I,EDE,EDE-S,EDE-SS}
			\end{axis}
			\end{tikzpicture}
			\caption{}
			\label{fig:binary_time}
		\end{subfigure}
		\caption{Summary results for Binary Strings (Section \ref{sec:binary_strings}).
		\subref{fig:binary_dg} The number of reactions ($|E(N)|$) and molecules ($|V(N)|$) in the resulting reaction network.
		\subref{fig:binary_calls} The number of yielded direct derivations during construction of each network.
		\subref{fig:binary_time} Runtime of each reaction network construction.
		\label{fig:binary_stats}}
\end{figure}
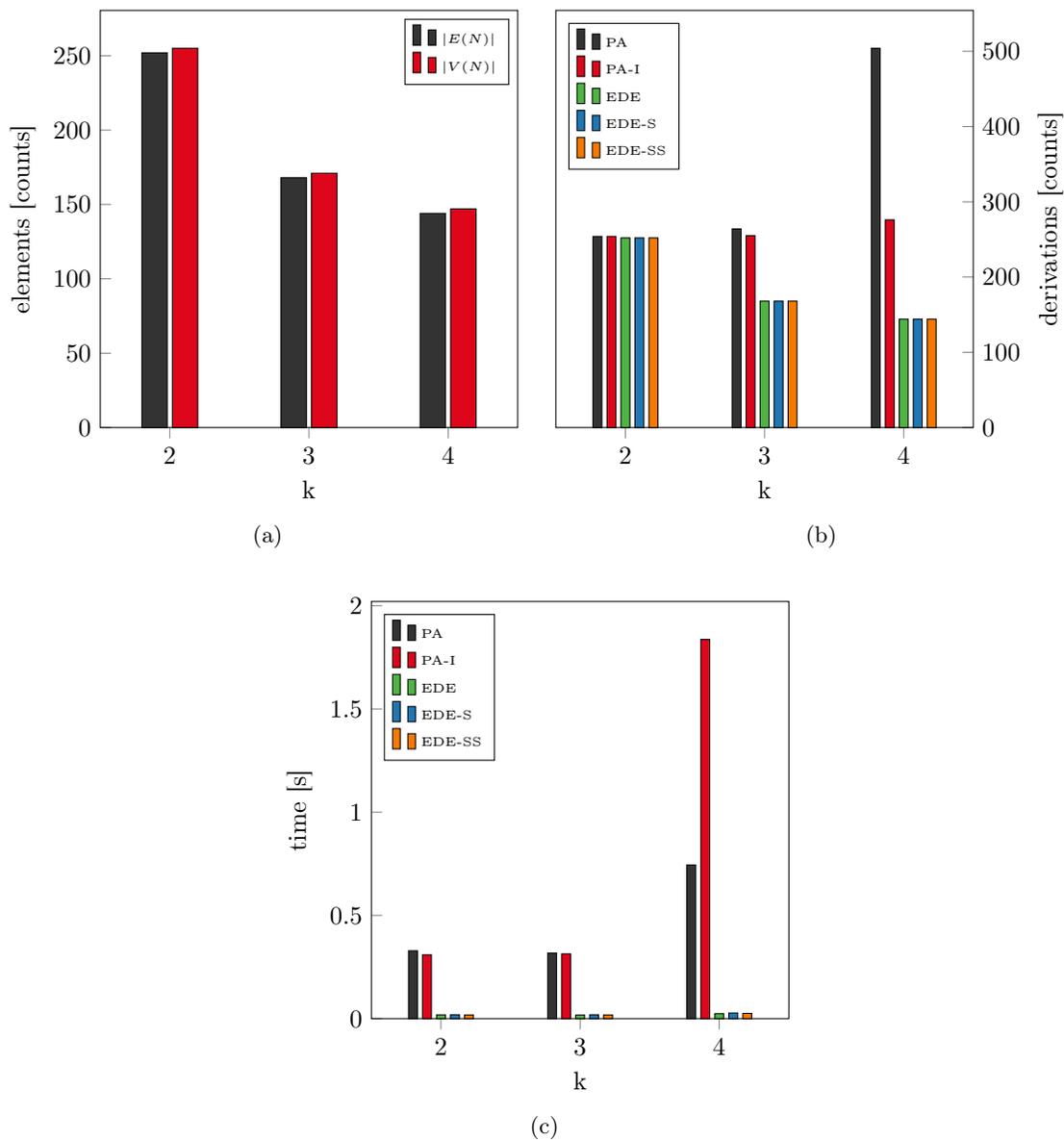

\subsection{Binary Trees}
\label{sec:binary_trees}

Complementing the first example, which focused on asymmetric rules and graphs, our second synthetic data experiment focuses on rules and graphs exhibiting various symmetries instead.
The experiment again consists of arranging the initial molecules into large structures. However, instead of simply growing in a linear fashion into a chain, we grow in all directions from a symmetric centerpiece with a four-cycle, shown in Figure~\ref{fig:binary_trees_square}.
Each carbon atom of the four-cycle then acts as a root of a binary tree of carbon atoms.
The tree is constructed by the rule depicted in Figure~\ref{fig:binary_trees_rule}, consisting of two symmetric components representing methane molecules (Figure~\ref{fig:binary_trees_star}) being attached to a leaf of a tree.
To prevent merging two trees, the methane molecules are fully specified in the rule.
Similarly, to prevent merging three methane molecules, the component representing the tree specifies an extra carbon atom.
The constraints imposed by this additional structure in the rule do not only allow us to control the growth of the binary trees, but in the case of the methane molecules also introduce symmetries into the graph components of the left side of the rule.

\begin{figure}
	\centering
	\begin{subfigure}[b]{0.25\textwidth}
		\centering
		\includegraphics[scale=0.7]{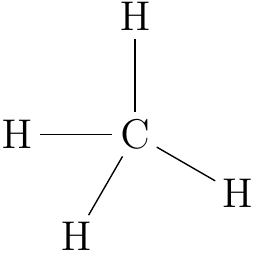}
		\caption{\label{fig:binary_trees_star}}
	\end{subfigure}
	\qquad
	\begin{subfigure}[b]{0.25\textwidth}
		\centering
		\includegraphics[scale=0.7]{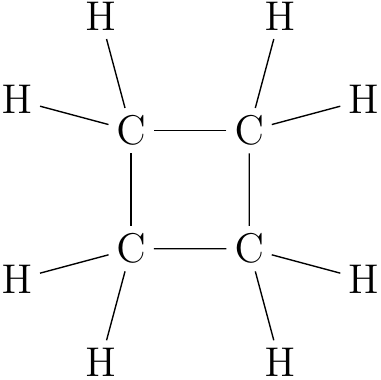}
		\caption{\label{fig:binary_trees_square}}
	\end{subfigure}

	\begin{subfigure}[b]{\textwidth}
		\centering
		\summaryRuleSpan[scale=0.7]{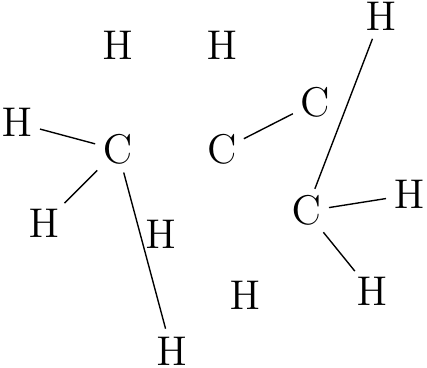}
		\caption{\label{fig:binary_trees_rule}}
	\end{subfigure}
	\caption{Graphs and rules used for Binary Trees (Section \ref{sec:binary_trees}).
		\subref{fig:binary_trees_star}-\subref{fig:binary_trees_square} The set of initial graphs.
		\subref{fig:binary_trees_rule} A rule that appends two children to the leaves (H atoms) of a binary tree.
		\label{fig:binary_trees_grammar}}
\end{figure}

We grow the reaction network for up to 10 iterations with a strict five minute timeout.
The results are summarized in the three graphs in Figure~\ref{fig:binary_trees_stats}, which show the size of the reaction network after each iteration, the number of derivations computed in each iteration, and the time taken to compute each iteration.
The benefit of pruning isomorphic matches, respectively isomorphic partially applied rules, is made clear by the number of computed derivations shown in Figure~\ref{fig:binary_trees_calls}.
The high number of isomorphic derivations enumerated by both \EDE{} and \PA{} is also reflected in the running time, leading to \EDE{} and \PA{} timing out during the $7\textsuperscript{th}$, respectively $6\textsuperscript{th}$, iteration.

While both \SEDE{} and \SSEDE{} successfully compute all $10$ iterations, \IPA{} timed out during
the $9\textsuperscript{th}$ derivation.
One can observe that the time taken by \IPA{} grows more sharply than the number of computed derivations, hinting again at the inefficiency of the pairwise isomorphism checks between the partially applied rules.
Another interesting observation is the discrepancy between run times and numbers of computed derivations between \SEDE{} and \SSEDE{}.
Due to the presence of complex symmetries, which are not captured by the generators, \SSEDE{} computes significantly more derivations in the $10\textsuperscript{th}$ iteration (\num{3e4}) compared to \SEDE{} (\num{7e3}).
\SSEDE{} still achieves better runtime, however, as \SEDE{} sufferes from the overhead of constructing the full symmetry group for each graph.

\begin{figure}
	\centering
	\begin{subfigure}[b]{.49\textwidth}
		\begin{tikzpicture}
		\begin{axis}[
		novo style,
		width=\linewidth,
		height=\linewidth,
		ylabel={elements [counts]},
		xlabel={iteration},
		ymax=100,
		legend pos=north west,
			legend cell align={left},
		]
		\plotfile{figures/binary_trees/binary_trees_dg.dat}
		\end{axis}
		\end{tikzpicture}
		\caption{}
		\label{fig:binary_trees_dg}
	\end{subfigure}
	\hfill
	\begin{subfigure}[b]{.49\textwidth}
		\begin{tikzpicture}
		\begin{axis}[
		novo style, 
		width=\linewidth,
		height=\linewidth,
		ylabel={derivations [counts]},
		xlabel={iteration},
		ymax=1200000,
		yticklabel pos=right,
		ytick pos=right,
			legend cell align={left},
		]
		\plotfile{figures/binary_trees/binary_trees_calls.dat}
		\end{axis}
		\end{tikzpicture}
		\caption{}
		\label{fig:binary_trees_calls}
	\end{subfigure}
	\\[.5cm]
	\begin{subfigure}[b]{.49\textwidth}
		\begin{tikzpicture}
		\begin{axis}[
		novo style, width=\linewidth,
		height=\linewidth,
		ylabel={time [\si{\second}]},
		xlabel={iteration},
			legend cell align={left},
		]
		\plotfile{figures/binary_trees/binary_trees_time.dat}
		\end{axis}
		\end{tikzpicture}
		\caption{}
		\label{fig:binary_trees_time}
	\end{subfigure}
	\caption{Summary results for Binary Trees (Section \ref{sec:binary_trees}).
		\subref{fig:binary_trees_dg} The number of reactions and molecules in the resulting reaction network.
		\subref{fig:binary_trees_calls} The number of yielded direct derivations during construction of each network.
		\subref{fig:binary_trees_time} Construction time for each reaction network.
		\label{fig:binary_trees_stats}}
\end{figure}
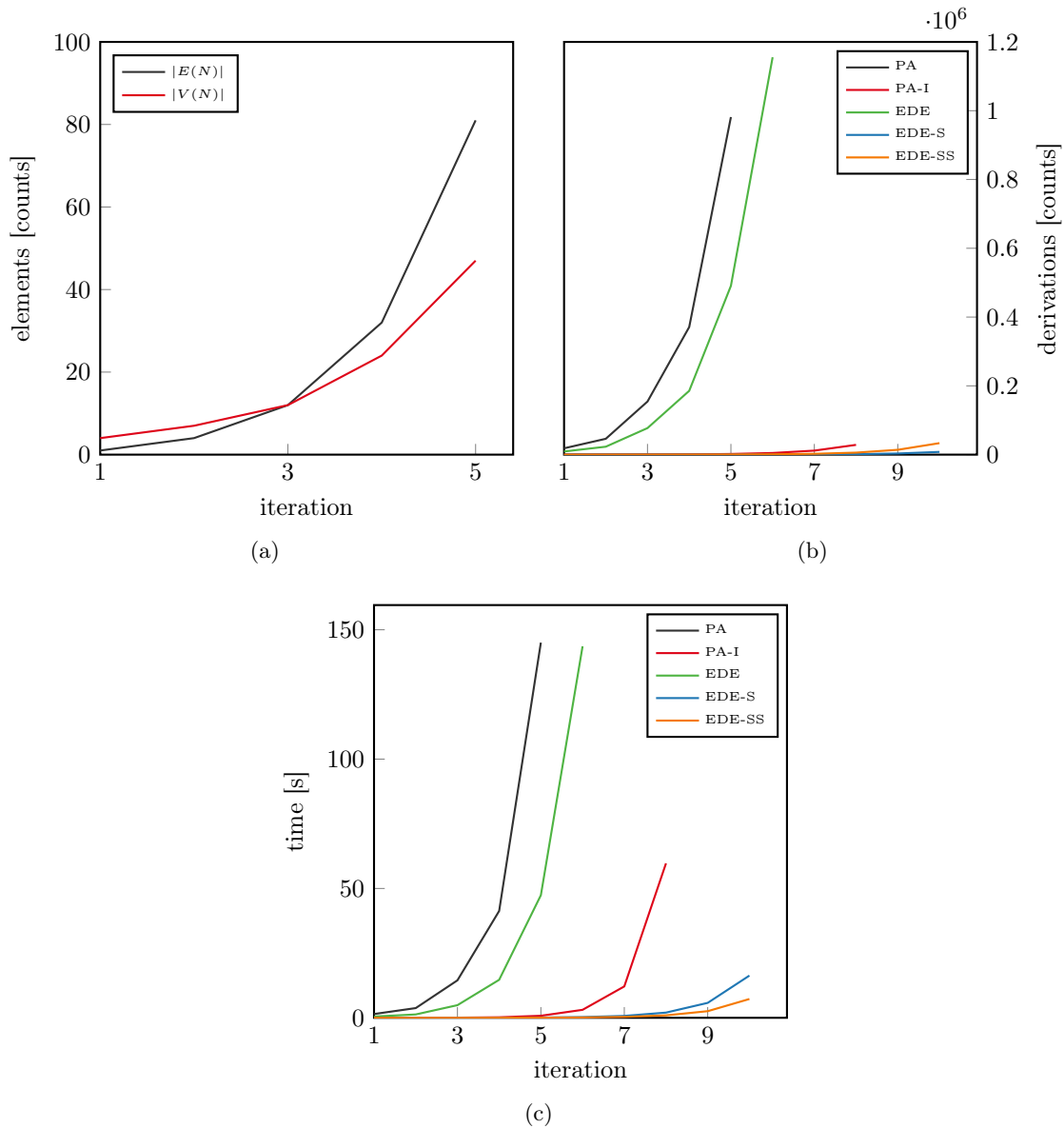

\subsection{The Formose Chemistry} \label{sec:formose}

For our first chemical experiment we select the well-known formose reaction describing the formation of sugars from formaldehyde.
The formose reaction was first modeled by graph transformation rules in~\cite{strat:14}.
Roughly speaking, the formation of sugars consists of polymerisation of formaldehyde molecules.
This is a two-step process where a sugar molecule first undergoes keto-enol isomerization (Figure~\ref{fig:formose_keto}) thereby preparing itself for the aldol addition (Figure~\ref{fig:formose_aldol}).
The aldol addition then binds a formaldehyde or another compatible molecule to the enolized sugar.
The inverse of both of the rules are also considered.

The performance comparison on the formose reaction gives an interesting baseline, as it is a chemical example which does not play into the strengths of our approach.
Indeed, only one out of the four rules has more than a single connected component in the left side graph.

Similarly to the binary string example (Section~\ref{sec:binary_strings}), we want to limit the reaction network size.
We therefore consider the formose reaction in multiple variations, with $n \in \{1, \dots, 13\}$ being the maximal number of carbon atoms in a single molecule.
The results for each $n$ are presented in Figure~\ref{fig:formose_stats}.

Due to the low number of connected components in the left side graphs of the rules, as well as the relatively small amount of symmetries, the number of computed derivations closely reflects the size of the reaction network and barely differs between all the variations of \PA{} and \EDE{}.
This is also reflected in the runtime of the algorithms, Figure~\ref{fig:formose_time}.
One can observe, however, that despite the similarity in the number of computed derivations, all variations of \EDE{} outperform both \PA{} variations, with the \SEDE{} trailing behind \EDE{} and \SSEDE{} due to the overhead from symmetry group construction.

\begin{figure}
	\centering
	\begin{subfigure}[b]{0.25\textwidth}
	\centering
	\includegraphics[scale=0.7]{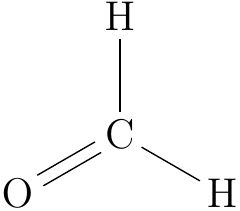}
	\caption{\label{fig:formose_formaldehyde}}
	\end{subfigure}
	\qquad
	\begin{subfigure}[b]{0.25\textwidth}
	\centering
	\includegraphics[scale=0.7]{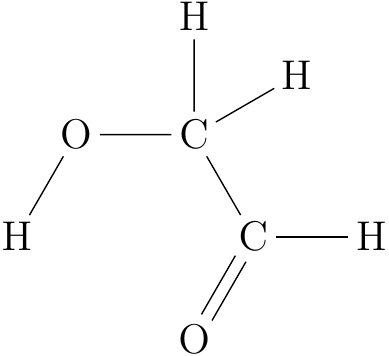}
	\caption{\label{fig:formose_glycolaldehyde}}
	\end{subfigure}

	\begin{subfigure}[b]{\textwidth}
	\centering
	\summaryRuleSpan[scale=0.7]{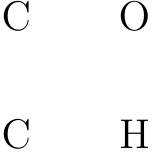}
	\caption{\label{fig:formose_keto}}
	\end{subfigure}

	\begin{subfigure}[b]{\textwidth}
	\centering
	\summaryRuleSpan[scale=0.7]{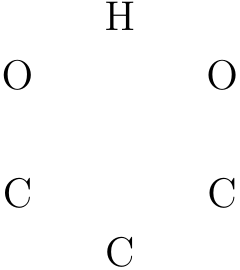}
	\caption{\label{fig:formose_aldol}}
	\end{subfigure}

	\caption{Graphs and rules used for formose (Section \ref{sec:formose}).
	\subref{fig:formose_formaldehyde} Formaldehyde.
	\subref{fig:formose_glycolaldehyde} Glycolaldehyde.
	\subref{fig:formose_keto}-\subref{fig:formose_aldol} Rules that together with their inverses model the formose chemistry.
	The rules model keto-enol-tautomerization and aldol addition, respectively.
	\label{fig:formose_grammar}}
\end{figure}

\begin{figure}
	\centering
	\begin{subfigure}[b]{.49\textwidth}
		\begin{tikzpicture}
		\begin{axis}[
		novo style,
		width=\linewidth,
		height=\linewidth,
		ylabel={elements [counts]},
		xlabel={$n$},
		ymax=25000,
		legend pos=north west,
			legend cell align={left},
		]
		\plotfile{figures/formose/formose_dg_stats.dat}
		\end{axis}
		\end{tikzpicture}
		\caption{}
		\label{fig:formose_dg}
	\end{subfigure}
	\hfill
	\begin{subfigure}[b]{.49\textwidth}
		\begin{tikzpicture}
		\begin{axis}[
		novo style, 
		width=\linewidth,
		height=\linewidth,
		ylabel={derivations [counts]},
		xlabel={$n$},
		ymax=26000,
		yticklabel pos=right,
		legend pos=north west,
		ytick pos=right,
			legend cell align={left},
		]
		\plotfile{figures/formose/formose_calls.dat}
		\end{axis}
		\end{tikzpicture}
		\caption{}
		\label{fig:formose_calls}
	\end{subfigure}
	\\[.5cm]
	\begin{subfigure}[b]{.49\textwidth}
		\begin{tikzpicture}
		\begin{axis}[
		novo style, width=\linewidth,
		height=\linewidth,
		ylabel={time [\si{\second}]},
		xlabel={$n$},
		ymax=1300,
		legend pos=north west,
			legend cell align={left},
		]
		\plotfile{figures/formose/formose_time.dat}
		\end{axis}
		\end{tikzpicture}
		\caption{}
		\label{fig:formose_time}
	\end{subfigure}
	\caption{Summary results for the formose chemistry (Section \ref{sec:formose}).
		The horizontal axis refers to the number of C atoms allowed in a connected component of the product graph. Derivations that produced molecules with more than $n$ carbon atoms were not considered. 
		\subref{fig:formose_dg} The number of reactions and molecules in the resulting reaction network.
		\subref{fig:formose_calls} The number of yielded direct derivations during construction of each network.
		\subref{fig:formose_time} Construction time for each reaction network.
		}
		\label{fig:formose_stats}
\end{figure}
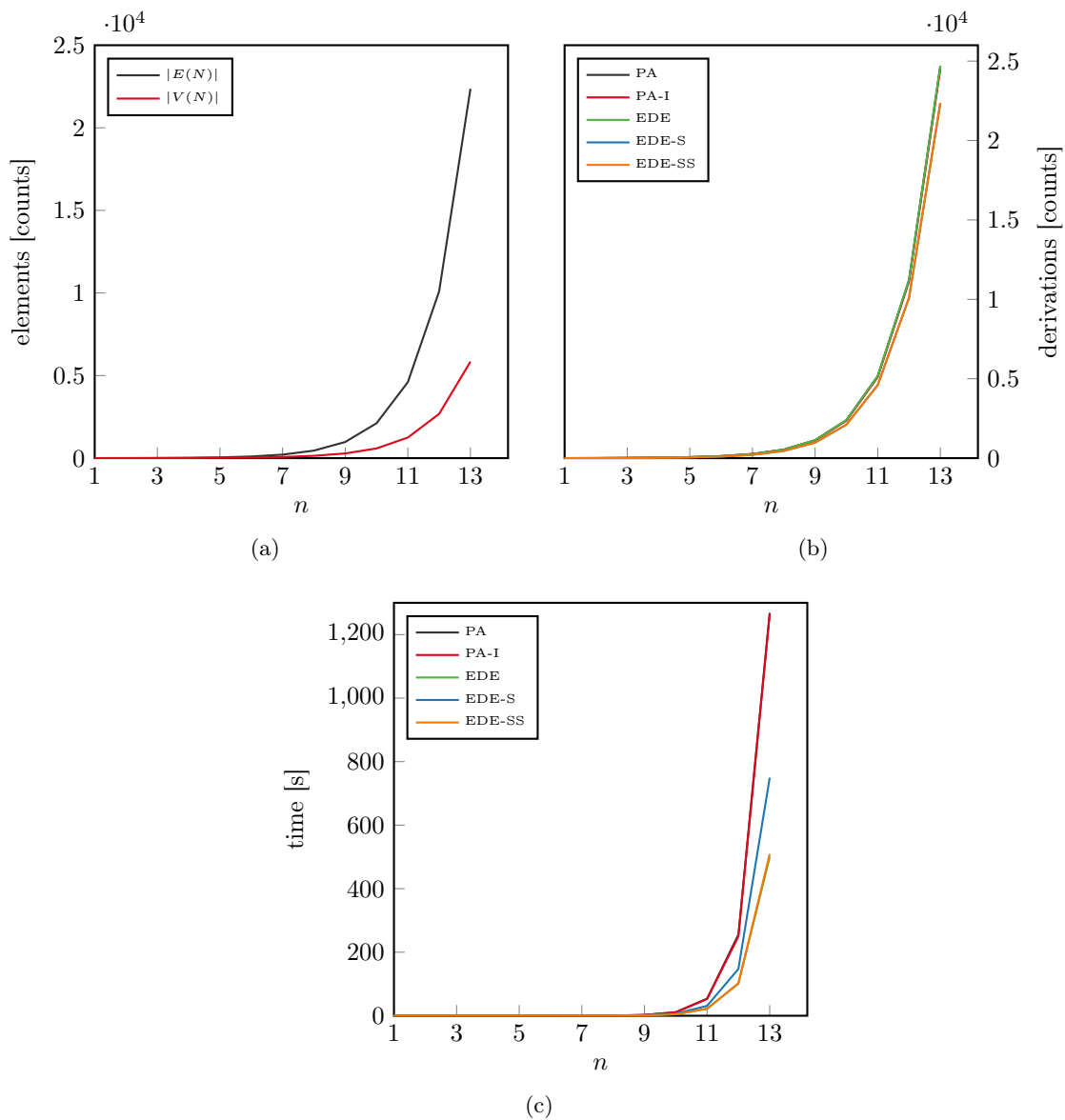

\subsection{Enzymatic Mechanisms} \label{sec:mcsa}
For the final experiment we examine enzymatic mechanisms modeled by graph transformation rules.
In simple terms, a reaction catalyzed by an enzyme can be viewed as a multi-step process, often known as a mechanism, which consists of elementary reaction steps converting the educts to products via a series of intermediate states, possibly temporarily modifying the enzyme itself in the process.
This view is adopted by a hand-curated database of enzymatic mechanisms, the Mechanism and Catalytic Site Atlas (M-CSA)~\cite{ribeiro2018mechanism}.
The M-CSA database has served~\cite{andersen2021graph} as a foundation for representing the reaction steps of enzymatic mechanisms by graph transformation rules, making the knowledge in M-CSA executable.
As a result, a total of $368$ different enzymatic mechanisms have been fully represented by graph transformation rules in~\cite{andersen2021graph}.
We use these enzymatic mechanism rules in our experiment.

In~\cite{andersen2021graph}, rules from all mechanisms are pooled together and used to propose new mechanisms for a given reaction.
In our experiments, to curtail the combinatorial explosion, we consider each of the $368$ mechanisms in isolation.
As such, each mechanism is only allowed to use the rules that have been derived from the mechanism itself.
To further control the reaction networks growth, we only allow molecules of up to $100$ atoms to appear on the product side and impose a runtime limit of 30 minutes per mechanism.
Under these conditions there are eight instances in which \PA{} cannot finish, twice due to the time limit and six times due to running out of available memory, limited to \SI{32}{\giga\byte} by the hardware of the machine used for the experiments.
Among all variations of \EDE{}, three instances fail to finish, two due to time and one due to memory.
In the following, we focus on the $360$ enzymatic mechanisms for which all configurations of both algorithms finish within the allocated resource limits.

Table~\ref{tab:mcsa_relative_all} lists runtime statistics of all tested algorithms on all enzymatic mechanisms as a relative percentage of the running time of \PA{}.
In the vast majority of cases, the produced reaction network is relatively small, and all algorithms finish the computation within a second.
The run time of \PA{} exceeds one second in only $20$ instances.
Irrespective of the total run time, all three variations of \EDE{} outperform both \PA{} and \IPA{} on almost all enzymatic mechanisms considered.
The exception being \SEDE{}, which performs slower than \PA{} on a few examples.
All such instances are on the small reaction networks, the difference between \PA{} and \SEDE{} being in the order of hundredths of a second.

On the other hand, Table~\ref{tab:mcsa_relative} lists running times as a relative percentage of the running time of \PA{} on the $20$ mechanisms that take more than a second to compute.
This comparison makes it clear that all variations of \EDE{} perform on average four to almost eight times faster than \PA{}, \SSEDE{} finishing consistently at least twice as fast
for reaction networks of sufficient complexity.

Table~\ref{tab:mcsa_relative} suggest \SSEDE{} as the fastest alternative.
This result is not surprising as many of the symmetries found in molecules are relatively simple,
such as symmetries between hydrogen atoms bonded to the same carbon atom, and can therefore be captured by a single generator.
More complex symmetries resulting from a composition of generators are far more rare, making the time saved by not constructing the symmetry groups overshadow the time saved by avoiding enumeration of all matches isomorphic by means of complex symmetries.

Finally, we highlight the $20$ enzymatic mechanisms that take the longest time to compute for \PA{}.
The performance of all algorithms as a relative percentage of the \PA{} run time is given individually for each of the $20$ mechanisms in Figure~\ref{fig:mcsa_rel_time_all}.
To avoid the issue of scale, the same data for only variations of \EDE{} is given in Figure~\ref{fig:mcsa_rel_time}.
One can remark that the performance varies greatly from one mechanism to the other.
An interesting case is the mechanism \texttt{131\_1}, which contains several rules with at least seven connected components in the left side graph.
As intended, \EDE{} has little trouble with applications of such rules, while the pairwise isomorphism checking in \IPA{} suffers from the exponential explosion in the number of partially applied rules.

\begin{table}
\centering
\begin{tabular}{lS[table-format=3.2]S[table-format=3.2]S[table-format=3.2]}
\toprule
			&        {mean} &         {min} &          {max} \\
\midrule
\EDE      &   27.39 &    0.44 &    90.55 \\
\SEDE    &   24.35 &    0.05 &   106.12 \\
\SSEDE   &   17.89 &    0.04 &    86.55 \\
\PA      &  100.00 &  100.00 &   100.00 \\
\IPA    &   96.13 &   13.55 &  1501.56 \\
\bottomrule
\end{tabular}
\caption{Average, minimum, and maximum recorded runtime across all enzymatic mechanisms for which all versions of \EDE{} and \PA{} finished.
All values are given as a percentage of the plain \PA{} runtime.
\label{tab:mcsa_relative_all}}
\end{table}

\begin{table}
\centering
\begin{tabular}{lS[table-format=3.2]S[table-format=3.2]S[table-format=3.2]}
\toprule
{} &        {mean} &         {min} &          {max} \\
\midrule
\EDE      &   25.06 &    0.44 &    65.87 \\
\SEDE    &   15.60 &    0.05 &    58.90 \\
\SSEDE   &   12.94 &    0.04 &    49.90 \\
\PA      &  100.00 &  100.00 &   100.00 \\
\IPA    &  233.37 &   18.56 &  1501.56 \\
\bottomrule
\end{tabular}

\caption{The same comparison as in Table~\ref{tab:mcsa_relative_all}, however, only for the 20 mechanisms with non-trivial runtime (more than one second for the plain \PA{}).
All values are given as a percentage of the plain \PA{} runtime.
\label{tab:mcsa_relative}}
\end{table}

\afterpage{
\clearpage
\begin{landscape}
\begin{figure}
	\centering
	\begin{subfigure}[b]{\linewidth}
		\begin{tikzpicture}
		\begin{axis}[
		ybar,
		novo bar style,
		ylabel={relative time [\si{\percent}]},
		xlabel={mechanism},
		xtick=data,
		bar width=2pt,
		height=.25\linewidth,
		tick label style={rotate=90},
			legend cell align={left},
		symbolic x coords={76\_1,131\_1,191\_1,243\_1,254\_2,331\_1,350\_1,512\_1,512\_2,526\_1,527\_1,557\_1,557\_2,582\_1,584\_1,593\_1,604\_1,659\_2,785\_1,889\_1},
		]
		\addplot[fill=papergrey] table[x=name, y=PA]{figures/mcsa/mcsa_rel_time_all.dat};
		\addplot[fill=paperred] table[x=name, y=PA-I]{figures/mcsa/mcsa_rel_time_all.dat};
		\addplot[fill=papergreen] table[x=name, y=EDE]{figures/mcsa/mcsa_rel_time_all.dat};
		\addplot[fill=paperblue] table[x=name, y=EDE-S]{figures/mcsa/mcsa_rel_time_all.dat};
		\addplot[fill=paperorange] table[x=name, y=EDE-SS]{figures/mcsa/mcsa_rel_time_all.dat};
		\legend{PA, PA-I,EDE,EDE-S,EDE-SS}
		\end{axis}
		\end{tikzpicture}
		\caption{}
		\label{fig:mcsa_rel_time_all}
	\end{subfigure}
	\\[.5cm]
	\begin{subfigure}[b]{\linewidth}
		\begin{tikzpicture}
		\begin{axis}[
		ybar,
		novo bar style,
		ylabel={relative time [\si{\percent}]},
		xlabel={mechanism},
		xtick=data,
		bar width=3pt,
		height=.25\linewidth,
		tick label style={rotate=90},
			legend cell align={left},
		symbolic x coords={76\_1,131\_1,191\_1,243\_1,254\_2,331\_1,350\_1,512\_1,512\_2,526\_1,527\_1,557\_1,557\_2,582\_1,584\_1,593\_1,604\_1,659\_2,785\_1,889\_1},
		]
		\addplot[fill=papergreen] table[x=name, y=EDE]{figures/mcsa/mcsa_rel_time.dat};
		\addplot[fill=paperblue] table[x=name, y=EDE-S]{figures/mcsa/mcsa_rel_time.dat};
		\addplot[fill=paperorange] table[x=name, y=EDE-SS]{figures/mcsa/mcsa_rel_time.dat};
		\legend{EDE,EDE-S,EDE-SS}
		\end{axis}
		\end{tikzpicture}
		\caption{}
		\label{fig:mcsa_rel_time}
	\end{subfigure}
	\caption{Summary results for enzymatic mechanisms (Section \ref{sec:mcsa}).
	The horizontal axes refer to the corresponding mechanism entry and proposal in the M-CSA database.
	\subref{fig:mcsa_rel_time_all} The relative runtime of all variants, scaled to the plain \PA{} variant.
	\subref{fig:mcsa_rel_time} The relative runtime of \EDE{} variants scaled to the plain \PA{} variant.
	\label{fig:mcsa_stats}}
\end{figure}
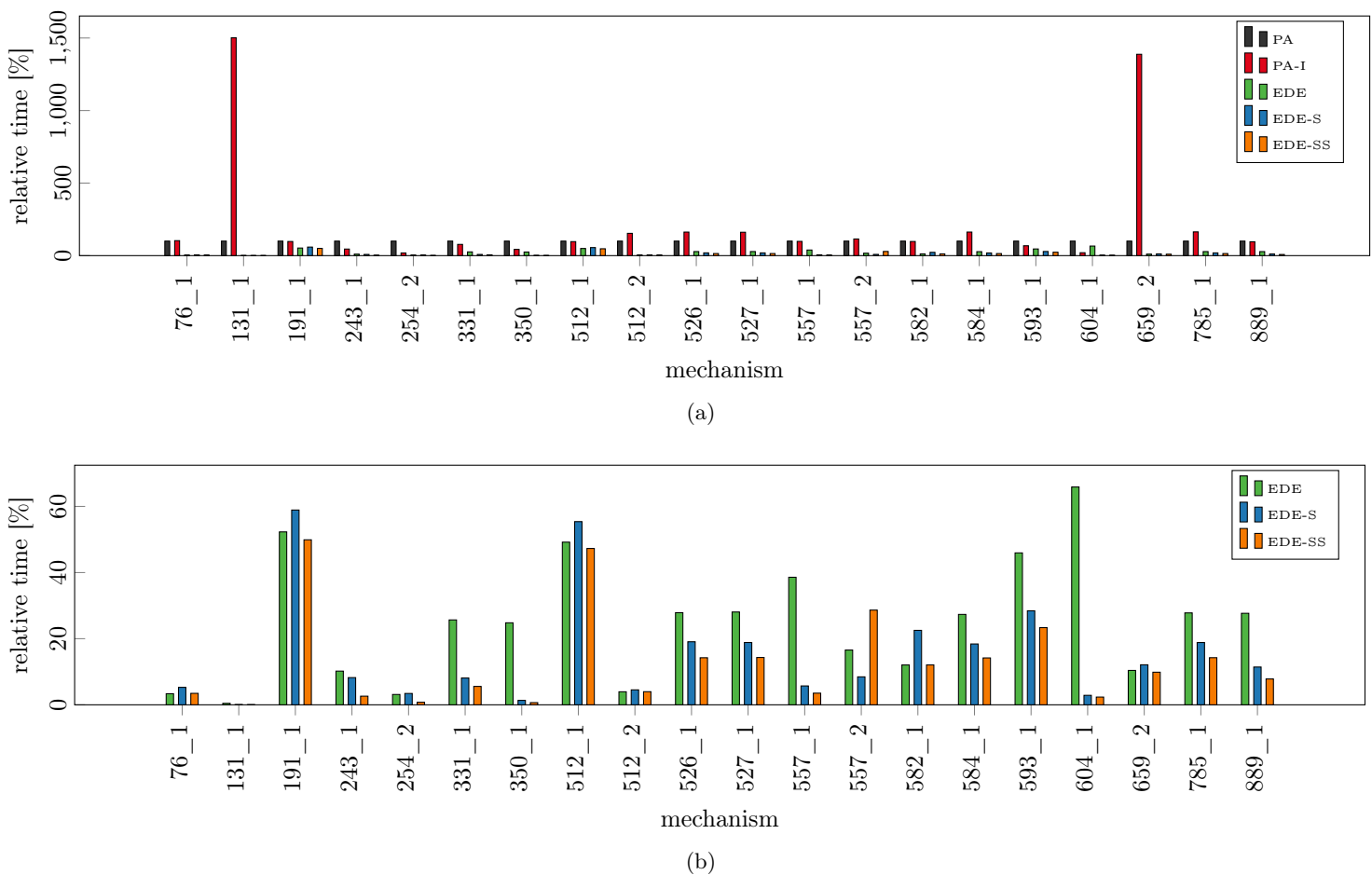
\end{landscape}
}

\section{Conclusion}
\label{sec:conclusion}

We design a new efficient algorithm for enumerating the matches of graph transformation rules in the application area of chemical reaction networks.
The proposed algorithm is designed to improve processor time when rules have multiple connected components on the left side. It also allows for trading off time for space by means of a pre-computation of the monomorphisms between rule components and potential host graphs.
While the algorithm is motivated specifically by rules with multiple connected components on the left side, our results show that the algorithm remains highly competitive even when rules with few connected components, e.g., the formose reaction in Section~\ref{sec:formose}, are utilized.

We further propose an alternative version of the algorithm which in addition to the validity of the matches also checks their uniqueness with respect to the resulting derivation isomorphisms.
This is achieved by the introduction of a total order on the partial matches based on the symmetry groups of the host molecules as well as the rules themselves in order to determine if a smaller partial match is available.
As a result, a partial match that does not lead to any non-isomorphic derivations can be discarded early, already during the enumeration, thereby saving resources.
We not only demonstrate that the new algorithm vastly outperforms existing methods,
which suffer from the combinatorial explosion in the number of derivations when checking for isomorphisms,
but also that not all symmetries are necessary to obtain a significant speedup.
This allows us to avoid the expensive operation of computing the automorphism groups of the graphs involved, using only generators of said groups instead.

\bibliographystyle{plain}
\bibliography{article}

\end{document}